\documentclass[onecolumn,a4paper,accepted=2022-05-30]{quantumarticle}
\pdfoutput=1
\usepackage[margin=1in]{geometry}
\usepackage{mathtools}
\usepackage{multicol}
\usepackage{stmaryrd}
\usepackage{kbordermatrix}
\usepackage{hyperref}
\usepackage[table]{xcolor}
\usepackage{amsmath}
\usepackage{amsfonts}
\usepackage{amsthm}
\usepackage{tikz}
\usepackage{doi}

\usepackage[customcolors]{hf-tikz}
\usepackage[ruled,lined]{algorithm2e}
\usepackage{qcircuit}
\usepackage{changepage}
\usepackage[outdir=./pictures/]{epstopdf}
\newtheorem{proposition}{Proposition}
\newcommand{\interp}[1]{\left\llbracket #1 \right\rrbracket}
\newcommand{\concat}[0]{\!::\!}
\newcommand{\ket}[1]{|#1\rangle}
\newcommand{\bra}[1]{\langle #1 |}

\definecolor{darkgreen}{rgb}{0.3,0.8,0.3}
\tikzset{style green/.style={
    set fill color=green!50!lime!60,
    set border color=white,
    draw opacity=0
  },
  style cyan/.style={
    set fill color=cyan!90!blue!60,
    set border color=white,
    draw opacity=0
  },
  style orange/.style={
    set fill color=orange!80!red!60,
    set border color=white,
    draw opacity=0
  },
  style purple/.style={
    set fill color=purple!80!red!20,
    set border color=white,
    draw opacity=0
  },
  hor/.style={
    above left offset={-0.15,0.31},
    below right offset={0.15,-0.125},
    #1
  },
  ver/.style={
    above left offset={-0.1,0.3},
    below right offset={0.26,-0.15},
    #1
  },
  vertab/.style={
    above left offset={-0.1,0.15},
    below right offset={0.26,-0.1},
    #1
  }
}

\title{Architecture aware compilation of quantum circuits via lazy synthesis}

\author[1]{Simon Martiel}
\email{simon.martiel@atos.net}
\author[1,2,3]{Timothée Goubault de Brugière}
\email{timothee.goubault-de-brugiere@loria.fr}
\affiliation[1]{Atos Quantum Lab.\\Les Clayes-sous-bois, France}
\affiliation[2]{Laboratoire de Recherche en Informatique (LRI),\\ Orsay, France}
\affiliation[3]{Laboratoire Lorrain de Recherche en Informatique et ses Applications (LORIA),\\Nancy, France}

\begin{document}

\maketitle

\begin{abstract}
    Qubit routing is a key problem for quantum circuit compilation. It consists in rewriting a quantum circuit by adding the least possible number of instructions to make the circuit 
    compliant with some architecture's connectivity constraints. Usually, this problem is tackled via either SWAP insertion techniques or re-synthesis of portions of the circuit using architecture
    aware synthesis algorithms. In this work, we propose a meta-heuristic that couples the iterative approach of SWAP insertion techniques with greedy architecture-aware synthesis routines. We propose two new compilation algorithms based on this meta-heuristic and compare their performances to state-of-the-art quantum circuit compilation techniques for several standard classes of quantum circuits and 
    show significant reduction in the entangling gate overhead due to compilation.
\end{abstract}

\section{Introduction}

Compilation is a key step in any software stack. Programs are often specified using a high-level programming language that allows the programmer to describe the manipulation of the processor's memory using abstract structures. This high-level description is then refined, sometimes in several stages, until it can be fully expressed as a sequence of low level instructions that can be executed by the processor.
Quantum programming makes no exception. In order to leverage the power of a quantum processor, one needs to compile high-level quantum programs into lower level sequences of quantum instructions.
This compilation step is particularly critical in the case of so called NISQ processors \cite{Preskill_2018}. In these settings, the quantum instructions are prone to errors and the quantum memory undergoes decoherence phenomena leading to quite large error rates. Consequently, there is a strong need for efficient heuristics to reduce the instruction count while still satisfying the architecture's constraints.

One of the most challenging problems in the field of compilation of quantum programs is the \emph{qubit routing} problem. Most quantum processors come with a limited chip connectivity, only allowing a (usually) small number of couplings between the different qubits. The input circuit should therefore be altered in order to only make use of the available interactions. This problem is traditionally tackled via the insertion of additional SWAP gates inside the circuit in order to move logical qubits from one physical qubit to the other \cite{hirata2011, 6560634, sabre, bka, Childs2019CircuitTF}. These techniques are inherently inefficient in the sense that they can only add gates to the compiled circuits and usually ignore the nature of the computation. 
Most of these algorithms lead to quite large SWAP/CNOT overheads when compared to the original circuit size. These overheads can be detrimental to the success rate of the algorithm. 

More recently, people started investigating the transverse approach of synthesizing quantum circuits that are readily compliant with a given connectivity. This approach is usually restricted to a particular subset of quantum circuits such as linear reversible operators \cite{kissinger2019cnot,de2020quantum} or phase polynomials \cite{griend2020architectureaware, Amy_2020, Nash_2020}. These methods have been shown to usually outperform SWAP insertion techniques, at the cost of lacking universality. In particular, these techniques require to either formulate the quantum programs using a higher level data structure, or pre-process the input circuit to slice it up into trivial pieces and sub-circuits that fit the class of operator we can synthesize. In a recent work \cite{gheorghiu2020reducing}, Gheorghiu et al. attempted to couple ad hoc synthesis techniques for phase polynomials with different splitting techniques to circumvent H gates. Their approach tends to show that the performance of the resulting compiler greatly depends on the splitting technique.

In this work, we propose a framework that strictly generalizes standard SWAP insertion approaches via lazy architecture-aware synthesis of partial unitary operators in a target subgroup. We provide a general formulation of the framework that is agnostic in the target subgroup. We then show how SWAP insertion techniques can be described by picking the permutation group as target subgroup. Finally, we define two new routing algorithms by applying our framework to the group of Boolean linear reversible operators and Clifford operators.

Our meta-heuristic can be informally described as follows:
\begin{itemize}
    \item Pick a subgroup of unitary operators that are easy to represent classically. By easy we mean that their classical representation has a polynomial size in the number of qubits and can be efficiently updated for composition.
    \item Initialize a data structure representing the identity.
    \item Iterate over the input circuit:
    \begin{itemize}
        \item if the incoming gate belongs to the subgroup, update the current data structure with this gate,
        \item if not, figure out a way to synthesize a piece of the current data structure into a circuit such that one can safely insert the incoming gate in the output circuit.
    \end{itemize}
\end{itemize}

This paper is organized as follows. We start by formalizing the above succinctly described meta-algorithm using what we call the \emph{lazy synthesis} framework. Section \ref{sec:swaps} shows how a standard SWAP insertion algorithm from \cite{hirata2011} fits into this framework. We then extend this algorithm by using the group of linear Boolean reversible operators in section \ref{sec:cnots} and the Clifford group in section \ref{sec:cliffords}. Some benchmarks against standard classes of circuits are provided and discussed in section \ref{sec:benchmarks}, together with a comparison with recent works in general purpose compilation. Finally, we propose some possible extensions, and conclude in a last section.

\medskip

\section{The lazy synthesis framework}
In this section we present a general formulation of the \emph{lazy synthesis} meta-heuristic. 

\medskip
\noindent{\bf Notations:}
\begin{itemize}
    \item Circuits are words on a (potentially infinite) gate set. We use $\concat$ for concatenation, and $\varepsilon$ for the empty circuit.
    \item Given some gate $g$, we denote by $\tilde{g}$ its corresponding $n$-qubits unitary operator, and extend this notation to circuits. For instance, given a circuit $c = g_1\concat g_2$ as a word, the corresponding equation in $\mathcal{U}(2^n)$ is $\tilde{c} = \tilde{g_2}\;\cdot\;\tilde{g_1}$ where $\cdot$ stands for the standard linear operator composition.
        
\end{itemize}

To introduce our framework, we first need to introduce some conventions.
We will assume that the input circuit is a sequence of gates taken from a set $\mathcal{G}_{in}$, and that the output circuit should have gates in another gate set $\mathcal{G}_{out}$.
Here, we voluntarily use a quite broad notion of gate set. For instance, $\mathcal{G}_{out}$ could contain the exact same gates as $\mathcal{G}_{in}$ but with additional constraints, such as connectivity constraints.
We will also assume that we have access to some data structure $\mathcal{D} = \langle \mathcal{H}, \interp{.}, S, u, e \rangle$ representing a class of unitary operators, with the following constraints:
\begin{itemize}
    \item $\mathcal{H}$ is some set of classical descriptions. We will usually require these descriptions to be small (i.e. polynomial in the number of qubits and/or the number of input gates).
    \item $\interp . : \mathcal{H} \rightarrow U(2^n)$ is an interpretation of the descriptions in $\mathcal{H}$ as unitary operators.
    \item $S \subseteq \mathcal{G}_{in}$ is a subset of the input gate set. Our data structure $\mathcal{D}$ corresponds to the class of operators that can be implemented by circuits with gates from $S$.
    \item $u: \mathcal{H} \times S \rightarrow \mathcal{H}$ is an update function such that:
        $$ \interp{u(h, g)} = \tilde{g}.\interp{h} $$
        that is, $u$ is sound with respect to $\interp .$. Less formally $u$ updates $h$ into $u(g,h)$ by absorbing $g$ into $h$. We will usually require for $u$ to efficiently update $h$ (i.e. runs in polynomial time w.r.t. the size of $h$).
    \item $e: \mathcal{H} \times \overline{S} \rightarrow \mathcal{H} \times \mathcal{G}_{out}^* $, where $\overline{S}$ is the complement of $S$ in $\mathcal{G}_{in}$. The function $e$ is an extraction function such that:
        $$ h', c = e(h, g) \implies \tilde{g}.\interp{h} = \interp{h'}.\tilde{c} $$
Less formally, $e$ tells us how to commute $g$ with $h$ as the cost of updating $h$ into $h'$ and turning $g$ into a sub-circuit $c$. We will usually require $e$ to be efficient.
\end{itemize}

Equipped with such a data structure, we can describe our meta-heuristic as the simple recipe detailed in Algorithm \ref{alg:lazy_synth}.

\begin{algorithm}
    \caption{Lazy synthesis meta-heuristic}
    \SetKwInOut{Input}{input}
    \SetKwInOut{Output}{output}
    \Input{circuit $c_{in}$}
    \Output{a final operator $h$ and a circuit $c_{out}$}
        $h \gets Id$\;
        $c_{out} \gets \varepsilon$\;
        \For{$g$ in $c_{in}$}{
            \eIf{$g\in S$}{
                $h \gets u(h, g)$\;
            }{
                $h', c = e(h, g)$\;
                $h\gets h'$\;
                $c_{out} \gets c_{out}\concat c$\;
            }
        }
       \Return $h, c_{out}$\;
    \label{alg:lazy_synth}
\end{algorithm}

The main idea of the heuristic is to iteratively aggregate gates of $c_{in}$ in $h$ and $c_{out}$ while maintaining the invariant: $\tilde{c}_{in}[1..i] = \interp{h}\cdot \tilde{c}_{out}$. That is: after compiling gate $i$, the initial segment $c_{in}[1..i]$ is equivalent to the composition of the current output circuit $c_{out}$ followed by the current stored operator $h$.
It is easy to check the soundness of the algorithm using the expected properties of $u$ and $e$. The process of our meta-heuristic is illustrated in Fig~\ref{fig:process}.

\begin{figure}
    \centering
    \hspace*{-2cm}
    \includegraphics[scale=0.65]{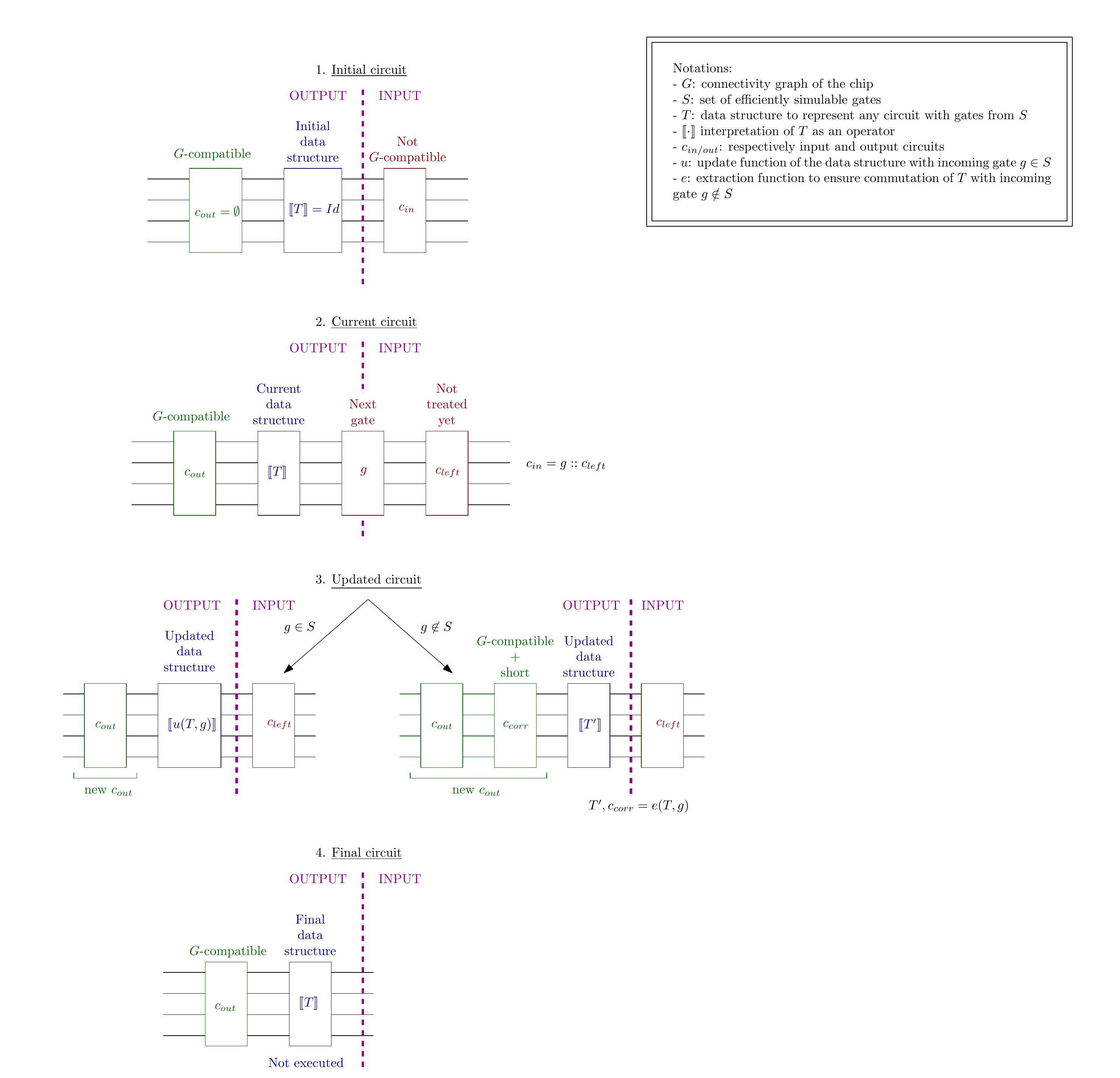}
    \caption{Illustration of Algorithm~\ref{alg:lazy_synth}. At any stage of the algorithm, we have the invariant $\tilde{c}_{left} \interp{T} \tilde{c}_{out}$ which is equal to the operator implemented by the input quantum circuit.}
    \label{fig:process}
\end{figure}

In other words, the gates in S are the ones we want to avoid executing by the quantum processor. As they belong to a group of efficiently simulable operators, our goal is to keep track classically of their action on the memory as long as possible with the use of our update function $u$. When a gate $g$ not belonging to S arises in the circuit, we try to minimize the quantity of extra gates needed to execute $g$ while keeping the functionality of the global operator. This is the goal of the extraction function $e$.

As you can notice, all the complexity of the heuristic lies in the implementation of the update and extraction functions $u$ and $e$. These functions will heavily rely on the underlying data structures.

In the next section, we show how to embed a SWAP insertion technique described in  \cite{hirata2011} into this framework. Later we will extend it to a broader set of operators to improve its performances, see Table~\ref{tab:my_label} for a summary of the different sets of operators considered in each section.

\section{Reformulation of swap insertion}\label{sec:swaps}

In \cite{hirata2011}, the authors propose a heuristic to iteratively rewrite a quantum circuit by inserting SWAP gates to route logical qubits.
In this approach, we will rely on the fact that elements in the group $S_n$ can be efficiently represented and manipulated. In order to represent an element $\sigma \in S_n$, we can simply store an array of integers $[\sigma(1), ..., \sigma(n)]$. Moreover, given the representations of two permutations $\sigma$ and $\pi$, the representation of $\sigma \circ \pi$ is simple to compute.

\medskip

\noindent{\bf Data structures.} We now describe how this algorithm is a particular case of our framework.

We first need to define $\mathcal{G}_{in}$, $\mathcal{G}_{out}$, and $S \subseteq \mathcal{G}_{in}$:
\begin{itemize}
    \item $\mathcal{G}_{in}$ contains any gate acting on at most $2$ qubits,
    \item $\mathcal{G}_{out}$ contains any gate acting on at most $2$ qubits and such that the gate is compatible with some connectivity graph $G$,
    \item finally $S = \{\textrm{SWAP}_{i, j}, i, j \in [n], i \neq j\}$ the set of all possible qubit SWAPs.
\end{itemize}
The classical data structure simply describes a qubit permutation: $\mathcal{D} = \langle S_n, \interp{.}, S, u, e \rangle$, where:
\begin{itemize}
    \item $S_n$ denotes the permutation group over $n$ elements, where $n=|V(G)|$ is the number of qubits.
    \item $\interp{.}$ trivially associates to a permutation the corresponding $n$-qubit unitary operator
    \item $u$ composes the current permutation with an incoming swap:
    $$u(\sigma, \textrm{SWAP}_{i, j}) = (i, j) \circ \sigma$$ 
\end{itemize}

\begin{table}[]
\begin{adjustwidth}{-1.5cm}{}
    \centering
    \begin{tabular}{|c|c|c|c|c|}
    \hline
        &$\mathcal{G}_{in}$ & $\mathcal{G}_{out}$ & subgroup & data structure\\
        \hline
        SWAP (Section 3) & $SU(2)$ + $SU(4)$ & SWAP + $SU(2)$ + $SU(4)$ & $S_n$ & arrays of indices\\
        \hline
        CNOT (Section 4) & CNOT + $SU(2)$ & CNOT + $SU(2)$ & $GL(n, \mathbb{F}_2)$ & invertible tables\\
        \hline
        Clifford (Section 5) & Clifford + Pauli rotations &$\{\operatorname{CNOT}, H, \sqrt{X}, R_Z\}$ & Clifford group & Tableaux\\
        \hline
    \end{tabular}
    \caption{This table summarizes the three instantiations of Algorithm \ref{alg:lazy_synth} for qubit permutations, linear boolean circuits, and Clifford circuits.}
    \label{tab:my_label}
\end{adjustwidth}
\end{table}

We now describe our extraction routine. Given some gate $g$ in the input circuit. If $g$ is such that $\sigma^{-1}(g)$ is compatible with $G$, we can simply use the fact that:
$$ g.\interp{\sigma} = \interp{\sigma}.\sigma^{-1}(g) $$ 
to set $e(\sigma, g) = \sigma, \sigma^{-1}(g)$. However, if $\sigma^{-1}(g)$ is not compatible with $G$, we need to produce a piece of $G$ compatible SWAP circuit $c_\pi$ implementing a permutation $\pi$ such that $\sigma'^{-1}(g)$ is compatible with $G$, with $\sigma' = \sigma \circ \pi^{-1}$. Then, we have that:

$$ g.\interp{\sigma} = \interp{\sigma\circ\pi^{-1}}.\sigma'^{-1}(g) .\interp{\pi} = \interp{\sigma\circ\pi^{-1}}.\sigma'^{-1}(g) .\tilde{c}_\pi$$

If we can produce such a circuit $c_\pi$, we can set $e(\sigma, g) = \sigma \circ \pi^{-1}, c_\pi::\sigma'^{-1}(g)$. We now describe how such a SWAP circuit is produced in Hirata et al. algorithm. 
Considering the fact that we need gate $\sigma'^{-1}(g) = (\pi\circ\sigma^{-1})(g)$ to be compatible with $G$, $\pi$ can be seen as a permutation bringing the qubits of $\sigma^{-1}(g)$ close to one another in $G$. Let $a, b$ be the pair of qubit on which $g$ acts and let $p=(\sigma^{-1}(a) = p_1, ..., p_k = \sigma^{-1}(b))$ be the shortest path from $\sigma^{-1}(a)$ to $\sigma^{-1}(b)$ in $G$.
The algorithm enumerates $k - 1$ permutations consisting in moving $\sigma^{-1}(a)$ toward $\sigma^{-1}(b)$ along $p$ and vice-versa until they meet somewhere along an edge of $p$.
For each of these permutations, the algorithm is called recursively for the next $w$ entangling gates, and the permutation leading to the lowest SWAP overhead is picked and committed to the output circuit, thus producing $c_\pi$. Figure \ref{fig:swaps} gives such an example of permutation enumeration. The general structure of such a recursive search is described in Appendix \ref{sec:rec_search}.
As expected, the performances of this algorithm heavily depend on the recursion depth parameter $w$.

The overall worst case complexity of this algorithm is $O(m n^{1 + w})$, with $m$ the number of entangling gates and $n$ the number of qubits, and neglecting the pre-computing of shortest-paths. Appendix \ref{app:example_swap} presents a step by step overview of the execution of this algorithm on a simple circuit for a LNN architecture. 

\begin{figure}[h]

\begin{tabular}{cccc}
(a)~~~\begin{tikzpicture}
\draw (0, 0) node[draw, rectangle, inner sep=1pt](n0){0};
\draw (1, 0) node[draw, circle, inner sep=1pt](n1){1};
\draw (2, 0) node[draw, circle, inner sep=1pt](n2){2};
\draw (0, -1) node[draw, circle, inner sep=1pt](n3){3};
\draw (1, -1) node[draw, circle, inner sep=1pt](n4){4};
\draw (2, -1) node[draw, circle, inner sep=1pt](n5){5};
\draw (0, -2) node[draw, circle, inner sep=1pt](n6){6};
\draw (1, -2) node[draw, rectangle, inner sep=1pt](n7){7};
\draw (2, -2) node[draw, circle, inner sep=1pt](n8){8};
\draw[thick, dashed] (n0) -- (n1);
\draw (n0) -- (n3);
\draw (n1) -- (n2);
\draw[thick, dashed] (n1) -- (n4);
\draw (n2) -- (n5);
\draw (n5) -- (n8);
\draw (n5) -- (n4);
\draw (n8) -- (n7);
\draw (n7) -- (n6);
\draw[thick, dashed] (n7) -- (n4);
\draw (n6) -- (n3);
\draw (n3) -- (n4);
\end{tikzpicture} & (b)~~~\begin{tikzpicture}
\draw (0, 0) node[draw, rectangle, inner sep=1pt](n0){0};
\draw (1, 0) node[draw, rectangle, inner sep=1pt](n1){7};
\draw (2, 0) node[draw, circle, inner sep=1pt](n2){2};
\draw (0, -1) node[draw, circle, inner sep=1pt](n3){3};
\draw (1, -1) node[draw, circle, inner sep=1pt](n4){1};
\draw (2, -1) node[draw, circle, inner sep=1pt](n5){5};
\draw (0, -2) node[draw, circle, inner sep=1pt](n6){6};
\draw (1, -2) node[draw, circle, inner sep=1pt](n7){4};
\draw (2, -2) node[draw, circle, inner sep=1pt](n8){8};
\draw[thick, dashed] (n0) -- (n1);
\draw (n0) -- (n3);
\draw (n1) -- (n2);
\draw[thick, <->] (n1) -- (n4);
\draw (n2) -- (n5);
\draw (n5) -- (n8);
\draw (n5) -- (n4);
\draw (n8) -- (n7);
\draw (n7) -- (n6);
\draw[thick, <->] (n7) -- (n4);
\draw (n6) -- (n3);
\draw (n3) -- (n4);
\end{tikzpicture} & (c)~~~\begin{tikzpicture}
\draw (0, 0) node[draw, circle, inner sep=1pt](n0){1};
\draw (1, 0) node[draw, rectangle, inner sep=1pt](n1){0};
\draw (2, 0) node[draw, circle, inner sep=1pt](n2){2};
\draw (0, -1) node[draw, circle, inner sep=1pt](n3){3};
\draw (1, -1) node[draw, rectangle, inner sep=1pt](n4){7};
\draw (2, -1) node[draw, circle, inner sep=1pt](n5){5};
\draw (0, -2) node[draw, circle, inner sep=1pt](n6){6};
\draw (1, -2) node[draw, circle, inner sep=1pt](n7){4};
\draw (2, -2) node[draw, circle, inner sep=1pt](n8){8};
\draw[thick, <->] (n0) -- (n1);
\draw (n0) -- (n3);
\draw (n1) -- (n2);
\draw[thick, dashed] (n1) -- (n4);
\draw (n2) -- (n5);
\draw (n5) -- (n8);
\draw (n5) -- (n4);
\draw (n8) -- (n7);
\draw (n7) -- (n6);
\draw[thick, <->] (n7) -- (n4);
\draw (n6) -- (n3);
\draw (n3) -- (n4);
\end{tikzpicture} & (d)~~~\begin{tikzpicture}
\draw (0, 0) node[draw, circle, inner sep=1pt](n0){1};
\draw (1, 0) node[draw, circle, inner sep=1pt](n1){4};
\draw (2, 0) node[draw, circle, inner sep=1pt](n2){2};
\draw (0, -1) node[draw, circle, inner sep=1pt](n3){3};
\draw (1, -1) node[draw, rectangle, inner sep=1pt](n4){0};
\draw (2, -1) node[draw, circle, inner sep=1pt](n5){5};
\draw (0, -2) node[draw, circle, inner sep=1pt](n6){6};
\draw (1, -2) node[draw, rectangle, inner sep=1pt](n7){7};
\draw (2, -2) node[draw, circle, inner sep=1pt](n8){8};
\draw[thick, <->] (n0) -- (n1);
\draw (n0) -- (n3);
\draw (n1) -- (n2);
\draw[thick, <->] (n1) -- (n4);
\draw (n2) -- (n5);
\draw (n5) -- (n8);
\draw (n5) -- (n4);
\draw (n8) -- (n7);
\draw (n7) -- (n6);
\draw[thick, dashed] (n7) -- (n4);
\draw (n6) -- (n3);
\draw (n3) -- (n4);
\end{tikzpicture}
\end{tabular}
\caption{Example of permutations explored by Hirata et al. algorithm. In this example, we need to apply a gate on qubits 0 and 7 while constrained on a $3\times 3$ grid architecture. We compute the shortest path between $0$ and $7$, $(0, 1, 4, 7)$ (dashed edges in (a)). We then explore three different permutations, each generated by $k-1 = 2$ inversions. These permutations are depicted in (b), (c), and (d). The dashed edges represent the place where the actual 2-qubit gate will take place. Among these three permutations, we pick the one that leads to the lowest SWAP overhead taking into account the next $w$ entangling gates for some fixed parameter $w$.}
\label{fig:swaps}
\end{figure}
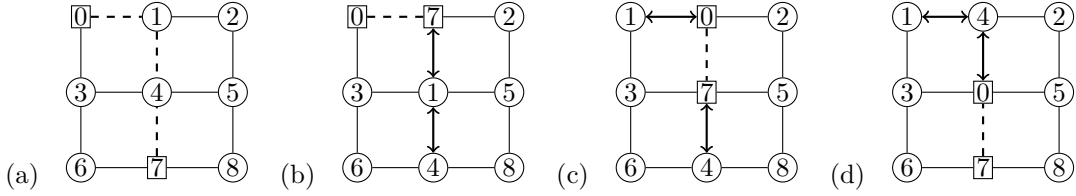

\section{Generalization to routing via lazy linear operator synthesis}\label{sec:cnots}

Now, using the lazy-synthesis framework to describe a SWAP insertion algorithm may seem a bit tedious and unnecessary. In this section, we show how, by extending our classical data structure, we can generalize Hirata et al. approach to outperform it in some settings.

\subsection{Data structures}\label{ssec:data_cnots}
We consider the set of reversible circuits over $n$ qubits comprising only CNOT gates. This set generates the entire set of reversible linear Boolean operators over $n$ variables, and in particular contains the set of all $n$ elements permutations.
This set has a lot of nice properties: it is easy to represent its elements via some $n\times n$ invertible Boolean tables, each row representing an output parity of the circuit \cite{amy2018controlled}. More precisely, given a linear reversible operator $A \in \mathbb{F}_2^{n \times n}$ acting on $n$ qubits at initial values $x=(x_0, x_1, ..., x_{n-1}), x_i \in \{0,1\}$, the logical value of the i-th qubit after execution of $A$ is given by 
\[ \alpha_0 x_0 \oplus \alpha_1x_1 \oplus ... \oplus \alpha_{n-1}x_{n-1}  \] 

where $\alpha = A[i,:]$ is the i-th row of $A$ and $\oplus$ stands for the XOR operation. Therefore we can keep track with a polynomially-sized structured of the action of CNOT circuits on the quantum memory.

Moreover, it is simple to update such tables via some row (resp. column) operations to accommodate for left (resp. right composition) of the operator by a CNOT \cite{patel2008optimal}. More generally, given an initial table $A$ and a linear reversible circuit implementing a table $B$, the updated table is given by $BA$.

\smallskip

\noindent{\bf Lazy linear synthesis}. Our gate sets are defined as follows:
\begin{itemize}
    \item $\mathcal{G}_{in}$ contains any 1-qubit gate and CNOT gates on arbitrary pairs of qubits, thus also including SWAPs,
    \item $\mathcal{G}_{out}$ contains any 1-qubit gate and CNOT gates compatible with some connectivity graph $G$,
    \item finally $S = \{CNOT_{i, j} | i, j \in V(G), i\neq j\}$ is the set of CNOT gates
\end{itemize}
The classical data structure describes reversible linear boolean operators over $n=|V(G)|$ qubits:
\begin{itemize}
    \item $\mathcal{H}$ is the set of invertible $n$ by $n$ boolean matrices,
    \item $\interp{.}$ trivially associates to a linear operator the corresponding $n$-qubit unitary operator,
    \item $u$ updates a table as expected with a matrix/matrix product:
        $$u(A, CNOT_{i, j}) = E_{i, j}.A$$
        where $E_{i, j}$ is the table representation of the operator $CNOT_{i,j}$ given by the identity matrix with one additional $1$ at row $j$, column $i$. In practice, given the simple structure of the $E_{i,j}$ operators, we recover the property that the action of a left-composition by a CNOT is equivalent to a row operation on $A$.
\end{itemize}

Given some incoming 1-qubit gate $g$ acting on qubit $q$ and some linear operator $A$, the behavior of our extraction routine relies on the following two properties: 
\begin{enumerate}
    \item if $A$ has shape: 
        \renewcommand{\kbldelim}{(}
        \renewcommand{\kbrdelim}{)}
        \begin{align}\label{ali:1}
            A = 
                \kbordermatrix{
                &   &       &   & q &    &        &\\
                &   &       &   & 0      &        & \\
                &   &  B' &   & \vdots &   &  B''  \\
                &   &       &   & 0      &        & \\
             q   & 0 &\cdots & 0 & 1      & 0 & \cdots & 0 \\
                &   &       &   & 0      &  &        & \\
                &   &   B'''  &   &\vdots  &   &       B'''' \\
                &&&&0      &       &
            }
        \end{align}
        then $A$ acts as the identity on qubit $q$. Consequently, any 1-qubit gate acting on qubit $q$ can commute with $A$.
    \item For any $B \in \mathbb{F}_2^{n \times n}$ invertible, we have the relation 
    \[ \interp{A} = \interp{ABB^{-1}} = \interp{AB} \cdot \interp{B^{-1}}. \]
    This means that if we add a linear reversible circuit implementing $B^{-1}$ to our current circuit, then to preserve the functionality of our quantum circuit the classical representation of the qubits is updated by $AB$.
\end{enumerate}

One can always find an operator $B$ such that $AB$ has the shape given by Eq.~\eqref{ali:1}. Given such an operator $B$, we have 
\[ \tilde{g} \cdot \interp{A} \underset{\text{property }2}{=} \tilde{g} \cdot \interp{AB} \cdot \interp{B^{-1}} \underset{\text{property }1}{=} \interp{AB} \cdot \tilde{g} \cdot \interp{B^{-1}}. \]

Hence, we define our extraction function $e$ as:
$$ e(A, g) = (AB, c\concat g) $$

where $B$ is such that $AB$ satisfies Eq.~\eqref{ali:1} and $c$ is a $G$-compatible circuit implementing $B^{-1}$.

In fact, we can slightly relax the structure of $AB$ and apply $g$ on a qubit different than qubit $q$. Indeed, considering another qubit $q' \neq q$ and writing $S_{q,q'}$ the Boolean linear operator associated to the swapping operator of qubits $q$ and $q'$, we have
\begin{equation} \tilde{g} \cdot \interp{A} = \interp{AB} \cdot \tilde{g} \cdot \interp{B^{-1}} = \interp{AB} \cdot \tilde{g} \cdot \interp{S_{q,q'}} \cdot \interp{S_{q,q'}} \cdot \interp{B^{-1}} = \interp{A(BS_{q,q'})} \cdot \tilde{g}' \cdot \interp{(BS_{q,q'})^{-1}} \label{swapping} \end{equation}

where $g'$ is the gate $g$ executed on qubit $q'$. In other words, as long as $A$ has shape $(1)$ up to a permutation of the columns, one can still apply gate $g$ on the qubit $q'$ for which $A[:,q'] = e_q$.

Our goal now is to find a suitable operator $B$ such that $c$ is the smallest possible. We provide a heuristic to construct such a circuit.

\subsection{Partial synthesis routine}
In order to simplify the description of our heuristic, we can first remark that the shape $(1)$ that we would like to achieve is stable under taking inverse. That is, finding $B$ such that $AB$ has shape $(1)$ is equivalent to finding $B^{-1}$ such that $B^{-1}A^{-1}$ has shape $(1)$. So instead of working on the columns of $A$ we can work on the rows of $A^{-1}$ and directly compute a quantum circuit for $B^{-1}$. Notably, due to Eq.~\eqref{swapping}, the freedom we have in the choice of the column for reducing $A$ to shape $(1)$ is now a freedom in the choice of the row of $A^{-1}$.

Given some incoming 1-qubit gate acting on qubit $q$, our heuristic works in two stages: 
\begin{itemize}
    \item We start by setting one row of $A^{-1}$ to $e_q^T$. 
    By definition of the inverse, the $q$-th row of $A$ 
    produces a bit vector describing which wire of the circuit should be fold using a fan-in CNOT (i.e. a cascade of CNOT gates that share the same target) onto one of them in order to produce $\{e_q\}$ on $A^{-1}$. By Eq.~\eqref{swapping} we can choose any of the wire $q'$ for which $A[q,q']=1$.
    \item After choosing a suitable qubit $q'$ and updating $A^{-1}$ accordingly, the $q$-th column of the operator can be zeroed by distributing the $q'$-th row 
    onto every row containing a non zero $q$-th component. This can be achieved using a single fan-out CNOT (i.e. a cascade of CNOT gates sharing the same control).
\end{itemize}

Hence, we simply need to be able to produce implementations of fan-in and fan-out CNOT gates that are compliant with our connectivity graph.

To perform this synthesis, we use a relaxed version of the method described in \cite{kissinger2019cnot}. This method relies on the construction of a Steiner tree, or rather an approximation of it. 

Given a connected graph $G$ and a subset of vertices $S \subset V(G)$, a \emph{Steiner tree with terminal vertices}
 $S$ is a subgraph $T$ of $G$ that $(i)$ is a tree, $(ii)$ contains all vertices in $S$, and is minimal (in terms of number of edges).
Given $G$ and $S\subset V(G)$, finding such a minimal tree is an NP-hard problem, even when $G$ is very structured
 (for instance it is hard to solve this problem for grids).
Nevertheless, it is possible to efficiently produce trees that are not too large compared to a Steiner tree.
In practice, we use the very standard algorithm of \cite{takahashi1990approximate}. This algorithm runs in time $O(kn^2)$, where $n$ is the number of vertices in $G$ and $k$ is the number of terminal vertices
, and achieves an approximation ratio of $2(1 - 1/k)$.

The idea behind the partial synthesis is the following:
\begin{itemize}
\item compute $y = e_q^T. A$,  $y = \{y_1, ..., y_k\}$
\item compute an (approximate) Steiner tree of the connectivity graph $G$, with terminal nodes $\{y_1, ..., y_k\}$
\item pick a terminal node $y_i$ and perform algorithm \ref{alg:fan_in}. This routine is a straightforward generalization of the nearest-neighbor implementation of a CNOT gate proposed in \cite{kutin2007computation}  (c.f their Figure 1) that is relaxed to leave intermediate wires in arbitrary states. It acts by pruning leaves of the tree while preserving the invariant that the leaves of the tree must be considered as control qubits for the rest of the fan-in synthesis. All CNOT gates used in the circuit are compliant with the tree's connectivity, making the circuit compliant with the qubits connectivity. Figure \ref{fig:ex_fan_in} gives an example of execution of this routine.
\end{itemize}

\begin{algorithm}[H]
    \caption{Fan-in along a tree}\label{alg:fan_in}
    \SetKwInOut{Input}{input}
    \SetKwInOut{Output}{output}
    \Input{tree $T$, vector $y$, vertex $root$}
    \Output{a circuit $c_{out}$}
    $c_{out} \gets \varepsilon$\;
    \While{$|T| > 1$}{
        $v \gets$ a leaf of $T$ that is not root\;
        $u \gets$ the only neighbor of $v$\;
        \If{$u \notin y$}{
            $c_{out} \gets c_{out}\concat CNOT(u, v)$\;
        }
        $c_{out} \gets c_{out}\concat CNOT(v, u)$\;
        $T.remove(v)$\;
    }
    \Return $c_{out}$\;
\end{algorithm}

Notice that intermediate wires may be left in a different state. Our only goal is to produce the correct parity $e_q$ on the root wire, and we take the liberty of freely changing the state of the intermediate wires.
The resulting circuit contains $2(l-1) -k$ CNOTs where $l$ is the size of the tree and $k$ is the number of terminal vertices (i.e. the Hamming weight of $y$), including the root of the tree.

Fan-outs are synthesized in a similar fashion, except terminal vertices are found by looking at lines of the updated operator $A'$ that have a non-zero $q$th component, and algorithm \ref{alg:fan_out} is used to produce a circuit.

\begin{algorithm}[H]
    \caption{Fan-out along a tree}\label{alg:fan_out}
    \SetKwInOut{Input}{input}
    \SetKwInOut{Output}{output}
    \SetKwComment{Comment}{//}{}
    \Input{tree $T$, vector $y$, vertex $root$}
    \Output{a circuit $c_{out}$}
    $c_{out} \gets \varepsilon$\;
    $Ones \gets y$\;
    $T' \gets$ a copy of $T$\;
    \Comment{Setting all the vertices of T to $1$}
    \While{$|T'| > 0$}{  
        $v \gets $ a leaf of $T'$\;
        $u \gets$ the only neighbor of $v$\;
        \If{$u \notin Ones$}{
            $c_{out} \gets c_{out}\concat CNOT(v, u)$\;
            $Ones.insert(u)$\;
        }
        $T'.remove(v)$\;
    }
    \Comment{Getting rid of all $1$s (except for root)}
    \While{$|T| > 1$}{ 
        $v \gets $ a leaf of $T$ thats not root\;
        $u \gets$ the only neighbor of $v$\;
        $c_{out} \gets c_{out}\concat CNOT(u, v)$\;
         $T.remove(v)$\;
    }
    \Return $c_{out}$\;
\end{algorithm}

This algorithm corresponds exactly to the fill-tree/empty-tree routine of \cite{kissinger2019cnot}, except that we work on the full hardware graph, 
and never have to restrict the structure of the Steiner-tree to a ``descending`` tree. This approach only works because we heavily rely on the fact that we are synthesizing a single row/column and 
thus allow ourselves to leave intermediate wires in arbitrary states.

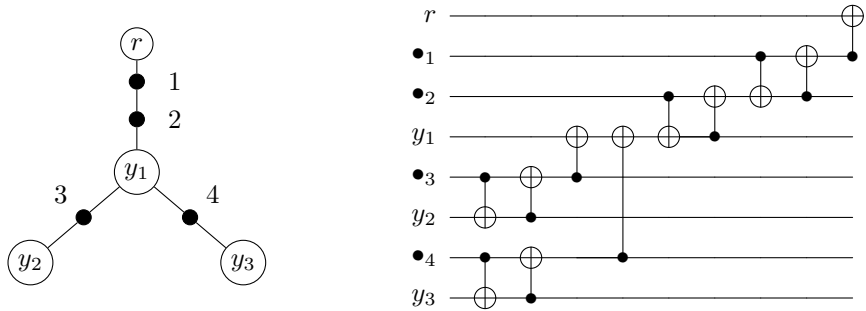
\begin{figure}[h]
\begin{center}
\begin{tikzpicture}
	\node[circle, draw, inner sep=2](root) at (0, 0){$r$};
	\node[circle, draw, inner sep=2, fill=black](r1_1) at  (0, -0.5){};
		\node() at  (0.5, -0.5){$1$};
	\node[circle, draw, inner sep=2, fill=black](r1_2) at  (0, -1){};
		\node() at  (0.5, -1){$2$};
	\node[circle, draw, inner sep=2](r1) at  (0, -1.7){$y_1$};
	\node[circle, draw, inner sep=2, fill=black](12_1) at  (-0.7, -2.3){};
		\node() at  (-1, -2){$3$};
	\node[circle, draw, inner sep=2](r2) at  (-1.4, -2.9){$y_2$};
	\node[circle, draw, inner sep=2, fill=black](13_1) at  (+0.7, -2.3){};
		\node() at  (+1, -2){$4$};
	\node[circle, draw, inner sep=2](r3) at  (+1.4, -2.9){$y_3$};
	\draw (root) -- (r1_1) -- (r1_2) -- (r1) -- (12_1) -- (r2);
	\draw (r1)
	      -- (13_1) -- (r3);	
	\draw (7,-1.5) node {
\Qcircuit @C=0.9em @R=0.7em {
\lstick{r}        &\qw&\qw&\qw&\qw&\qw&\qw&\qw&\qw&\targ&\\
\lstick{\bullet_1}&\qw&\qw&\qw&\qw&\qw&\qw&\ctrl{1}&\targ&\ctrl{-1}&\\
\lstick{\bullet_2}&\qw&\qw&\qw&\qw&\ctrl{1}&\targ&\targ&\ctrl{-1}&\qw&\\
\lstick{y_1}      &\qw&\qw&\targ&\targ&\targ&\ctrl{-1}\qw&\qw&\qw&\qw&\\
\lstick{\bullet_3}&\ctrl{1}& \targ&\ctrl{-1}&\qw&\qw&\qw&\qw&\qw&\qw&\\
\lstick{y_2}      &\targ&\ctrl{-1}&\qw&\qw&\qw&\qw&\qw&\qw&\qw&\\
\lstick{\bullet_4}&\ctrl{1}& \targ&\qw&\ctrl{-3}\qw&\qw&\qw&\qw&\qw&\qw&\\
\lstick{y_3}      &\targ&\ctrl{-1}&\qw&\qw&\qw&\qw&\qw&\qw&\qw&
}
};
\end{tikzpicture}
\end{center}

    \caption{Example of a tree and the corresponding fan-in CNOT circuit generated by algorithm \ref{alg:fan_in}. The terminal vertices are circled. Intermediate vertices are represented as $\bullet$. Notice that this routine can be improved in order to reduce the depth of the fan-in gate. In this work we decided to focus on CNOT count and thus did not insist on these lower level optimizations.}\label{fig:ex_fan_in}
\end{figure}

Both of these routines are quite close to the one used in \cite{kissinger2019cnot}, except that we allow ourselves to be sloppier in the process, and leaving any intermediate qubit in a dirty state, instead of having to preserve invariants when implementing the fan-in/fan-outs.
Appendix \ref{app:example_linear} gives a step by step overview of the execution of this algorithm on a simple circuit.

\subsection{Further optimizations}\label{ssec:ext_cnots}

In practice we improve the algorithm using two independent optimizations.

\smallskip

\noindent{\bf Dealing with phase gates.} It is unnecessary to zero a column of our current linear operator if we just need to insert a phase gate (i.e. a diagonal gate). 
Indeed, since the gate is diagonal, and assuming it is executed on qubit $q$, it is well-known that the gate commutes with any CNOT whose target is not $q$. So the diagonal gate will commute with the subsequent fan-out because one can check that the CNOT gates of the fan-out only use the qubit on which the diagonal gate is executed as a control. Hence, this fan-out can be omitted, thus approximately halving the number of required CNOT gates.

\medskip

\noindent{\bf Recursive search of finite depth.} As mentioned at the end of section \ref{ssec:data_cnots}, we can synthesize our operator $B$ up to some column permutation. This gives us some freedom to perform some optimizations when picking the qubit that will effectively receive the incoming 1-qubit gate.
 To leverage this freedom, we can adopt the same
strategy as in Hirata et al. SWAP insertion algorithm.
In practice, given an incoming gate $g$ acting on qubit $q$, we:
\begin{itemize}
    \item compute the set $y$ of rows of $A^{-1}$ that need to interact in the fan-in CNOT,
    \item generate a Steiner tree with terminal vertices $y$,
    \item branch over all choices of $y_i \in y$ to receive gate $g$
\end{itemize}

Notice that this boils down to trying all possible terminal vertices as root vertices in algorithm \ref{alg:fan_in}. We then perform a recursive search as described in Appendix \ref{sec:rec_search}.

Overall, including a recursive search of depth $w$, the worst case time complexity of our algorithm grows as $O(m n^{3 + w})$ where $m$ is the number of $1$-qubit gates, $n$ the number of qubits in the target architecture. Notice that the runtime is linear in the input circuit's size, but grow exponentially in the depth of the recursive search.

\medskip

\noindent{\bf Dealing with the final operator.} In the general case, the final linear operator in our classical data structure is not trivial. In a general compilation setting, this is not much of an issue, for two reasons:
\begin{itemize}
    \item in the setting where we might have a follow up circuit to compile, one can initialize the linear operator for the next compilation round to the final operator of the previous round,
    \item if we just finished compiling the final portion of our full quantum algorithm, one can always fix the sampled data in order to classically emulate the final linear operator. This operation boils down to inverting a simple linear system over $\mathbb{F}_2$.
\end{itemize}
Moreover, in most NISQ applications, the sampling directive executed at the end of a quantum circuit are here to estimate the expected value of some Hermitian operator $H$. Most of the time, this operator is specified in the Pauli basis. Thus, it is enough to compute a new Hermitian operator $A^{-1}HA$ such that sampling this operator at the end of the compiled circuit is equivalent to sampling the original operator at the end of the input circuit, and this new operator has the same number of terms as the original operator:
$$ \bra{0}C_{in}^\dagger HC_{in}\ket{0} = \bra{0}C_{out}^\dagger \left(A^{-1} H A \right) C_{out}  \ket{0} $$ 

In fact, this property is true for a larger subgroup: the Clifford group, which is tackled in the following section. The fixing procedure for the sampling and observable cases are detailed in Appendix \ref{sec:opt} in the more general case of Clifford operators.

\section{Generalization to routing via lazy synthesis of Clifford operators}\label{sec:cliffords}

We now further extend the previous approaches to lazy synthesis of elements of the Clifford group.

\subsection{Clifford group, Pauli rotations, and tableaux}\label{subsec:clifford}
The Clifford group, $\mathcal{C}_n$, is a natural extension of the class of reversible linear Boolean operators. This group is defined as the largest subgroup of the unitary group that stabilizes the group of Pauli operators $\mathcal{P}_n$:
\begin{align}
 \mathcal{C}_n = \{ U \in U(2^n),\ \forall P\in \mathcal{P}_n,\  U^\dagger P U \in  \mathcal{P}_n \} \label{eq:clifford}
\end{align}

Given a Pauli operator $P\in \mathcal{P}_n$ and a real angle $\theta\in\mathbb{R}$, we define the Pauli rotation $R_P(\theta)$ as:
$$R_P(\theta) = \cos(\theta/2)\mathbb{I} - i \sin(\theta/2)P$$

 The conjugation property \ref{eq:clifford} also applies to Pauli rotations, and not only Pauli operators. Hence, for any Pauli rotation of axis $P\in\mathcal{P}_n$ and any angle $\theta\in \mathbb{R}$, and any $U\in \mathcal{C}_n$:
$$ U^\dagger R_P(\theta)U = R_{U^\dagger P U}(\theta) = R_{P'}(s\cdot\theta)$$
for some Pauli operator $P'$ and some sign $s={\pm 1}$.

This conjugation relation between Clifford operators and Pauli rotations can be exploited as a generalized commutation relation:
$$ R_P(\theta)U = U R_{P'}(s\cdot\theta). $$ 
In fact, this relation can be used to normalize quantum circuits as sequences of non-Clifford Pauli rotations (i.e. Pauli rotations with angles $\neq k\pi/2$), followed by a final Clifford operator (a good example can be found in \cite{Litinski_2019}).

Elements of the Clifford group can be represented efficiently using data structures called \emph{tableaux} that specify how they act by conjugation over generators of the Pauli group \cite{Aaronson_2004, beaudrap2011linearized}. In practice, this means that we can implement a data structure $T$ (a tableau), representing a Clifford operator in $\mathcal{C}$ that:
\begin{itemize}
    \item can be easily updated $T \gets \tilde{g} \cdot T$ or $T \gets T\cdot \tilde{g}$ for some Clifford gate $g$,
    \item can be used to efficiently compute $P \mapsto T P T^\dagger$ for some $n$-qubits Pauli operator $P$, yielding another Pauli operator (and potentially a phase in $\pm 1$).
\end{itemize}
This is a key ingredient to prove the Gottesman-Knill theorem, stating that the execution of quantum circuits over Clifford gates can be efficiently simulated.

In our case, we will assume that we have access to such a data structure, without diving into the implementation details. In practice, we relied on the formulation presented by de Beaudrap in \cite{beaudrap2011linearized} which allows to easily invert tableaux, as well as perform the two operations described above.

Notice that most of this could also be efficiently done by storing a Clifford circuit. In that case, the update operation becomes trivial and the conjugation operation will consist in sequentially conjugating the input Pauli operator by each gate of the Clifford circuit.

\subsection{Lazy synthesis of Clifford operator}

In the following, we define the support of a Pauli operator $P$ as the set of qubits such that $P$ acts non-trivially on them. E.g if $P=I\otimes Z \otimes X \otimes I$, the the support of $P$ is the set $\{1, 2\}$ since $P$ acts as the identity on qubits $0$ and $3$. For ease of notations we will drop the $\otimes$ operators.

\medskip

\noindent{\bf Remark.} In the following subsection, we will use the following simple structure to implement a Pauli rotation $R_P(\theta)$. We can first reduce $P$ to a diagonal operator by conjugating it through a circuit composed of local Clifford gates. This circuit can be built by individually diagonalizing each component of the Pauli operator: 
\begin{itemize}
    \item if the operator acts as $X$ on qubit $i$, insert a $H$ gate on qubit $i$, 
    \item if it acts as $Y$ on qubit $i$, insert a $\sqrt{X} = R_X(\pi/2)$ on qubit $i$.
\end{itemize}

The resulting Pauli operator acts either as $Z$ or $I$ on each qubit. Using the identity $\textrm{CNOT} \cdot ZZ \cdot \textrm{CNOT} = IZ$, one can reduce the support of $P$ to a single qubit via conjugation by a circuit composed of $|P|-1$ CNOT gates sharing the same target qubit $q$. In fine, the resulting Clifford circuit $C$ verifies $R_P(\theta)= C^\dagger R_{Z_q}(\theta) C$. An example is given in Figure \ref{fig:pauli_reduction}. This reduction can be easily extended to take architecture into account by performing a fan-in CNOT along a Steiner tree with the support of the rotation as terminal vertices.

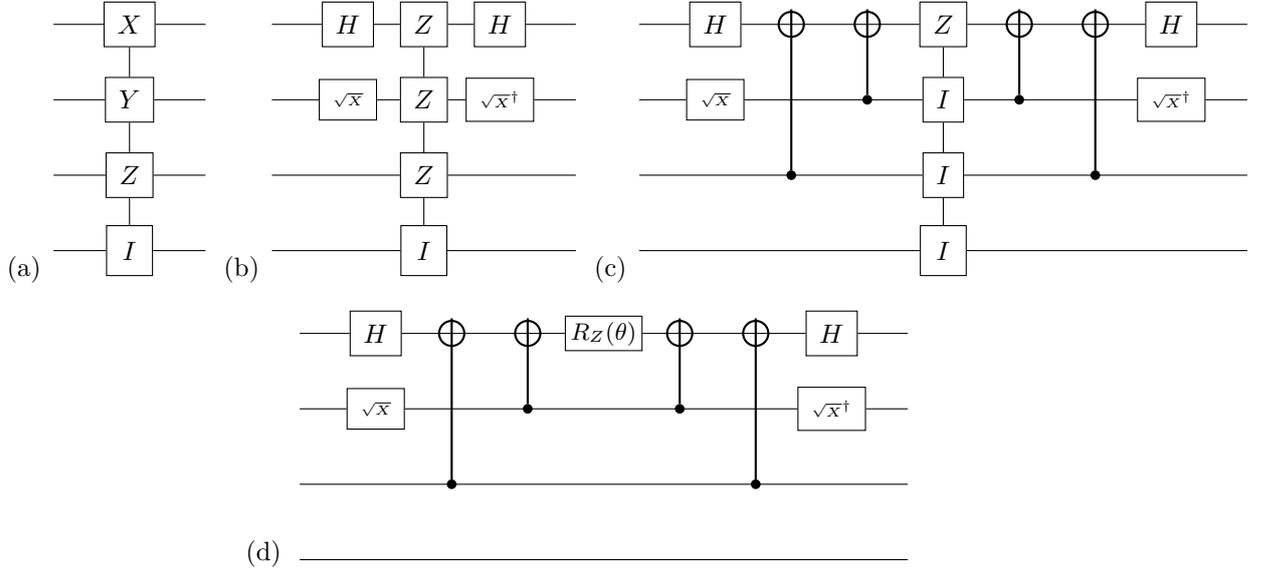
\begin{figure}[h]
    \centering
    
(a)~~\begin{tikzpicture}
\draw (0, 0)  node[rectangle, draw, inner sep=5pt](g1){$X$};
\draw (0, -1) node[rectangle, draw, inner sep=5pt](g2){$Y$};
\draw (0, -2) node[rectangle, draw, inner sep=5pt](g3){$Z$};
\draw (0, -3) node[rectangle, draw, inner sep=6pt](g4){$I$};
\draw (g1) -- (g2) -- (g3) -- (g4);

\draw (1, 0)  -- (g1);
\draw (1, -1)  -- (g2);
\draw (1, -2)  -- (g3);
\draw (1, -3)  -- (g4);

\draw (-1, 0)  -- (g1);
\draw (-1, -1)  -- (g2);
\draw (-1, -2)  -- (g3);
\draw (-1, -3)  -- (g4);
\end{tikzpicture}~~~(b)~~\begin{tikzpicture}
\draw (-1, 0)  node[rectangle, draw, inner sep=5pt](c11){$H$};
\draw (1, 0)  node[rectangle, draw, inner sep=5pt](c12){$H$};
\draw (0, 0)  node[rectangle, draw, inner sep=5pt](g1){$Z$};

\draw (-1, -1)  node[rectangle, draw, inner sep=5pt](c21){\tiny$\sqrt{X}$};
\draw (1, -1)  node[rectangle, draw, inner sep=5pt](c22){\tiny$\sqrt{X}^\dagger$};
\draw (0, -1) node[rectangle, draw, inner sep=5pt](g2){$Z$};

\draw (0, -2) node[rectangle, draw, inner sep=5pt](g3){$Z$};

\draw (0, -3) node[rectangle, draw, inner sep=6pt](g4){$I$};
\draw (g1) -- (g2) -- (g3) -- (g4);
\draw (c11) -- (g1) -- (c12);
\draw (c21) -- (g2) -- (c22);

\draw (2, 0)  -- (c12);
\draw (2, -1)  -- (c22);
\draw (2, -2)  -- (g3);
\draw (2, -3)  -- (g4);

\draw (-2, 0)  -- (c11);
\draw (-2, -1)  -- (c21);
\draw (-2, -2)  -- (g3);
\draw (-2, -3)  -- (g4);
\end{tikzpicture}~~~(c)~~\begin{tikzpicture}
\draw (-1, 0)  node[rectangle, draw, inner sep=5pt](c11){$H$};
\draw (5, 0)  node[rectangle, draw, inner sep=5pt](c12){$H$};

\draw (1, 0) node[circle, draw, thick](t1){};
\draw (1, -1) node[circle, draw, thick, fill=black, inner sep=1pt](c1){};
\draw[thick] (c1) -- (1, 0.2);

\draw (3, 0) node[circle, draw, thick](t1){};
\draw (3, -1) node[circle, draw, thick, fill=black, inner sep=1pt](c1){};
\draw[thick] (c1) -- (3, 0.2);

\draw (0, 0) node[circle, draw, thick](t2){};
\draw (0, -2) node[circle, draw, thick, fill=black, inner sep=1pt](c2){};
\draw[thick] (c2) -- (0, 0.2);

\draw (4, 0) node[circle, draw, thick](t2){};
\draw (4, -2) node[circle, draw, thick, fill=black, inner sep=1pt](c2){};
\draw[thick] (c2) -- (4, 0.2);

\draw (2, 0)  node[rectangle, draw, inner sep=5pt](g1){$Z$};

\draw (-1, -1)  node[rectangle, draw, inner sep=5pt](c21){\tiny$\sqrt{X}$};
\draw (5, -1)  node[rectangle, draw, inner sep=5pt](c22){\tiny$\sqrt{X}^\dagger$};
\draw (2, -1) node[rectangle, draw, inner sep=5pt](g2){$I$};

\draw (2, -2) node[rectangle, draw, inner sep=5pt](g3){$I$};

\draw (2, -3) node[rectangle, draw, inner sep=6pt](g4){$I$};
\draw (g1) -- (g2) -- (g3) -- (g4);
\draw (c11) -- (g1) -- (c12);
\draw (c21) -- (g2) -- (c22);

\draw (6, 0)  -- (c12);
\draw (6, -1)  -- (c22);
\draw (6, -2)  -- (g3);
\draw (6, -3)  -- (g4);

\draw (-2, 0)  -- (c11);
\draw (-2, -1)  -- (c21);
\draw (-2, -2)  -- (g3);
\draw (-2, -3)  -- (g4);
\end{tikzpicture}  
~\\
(d)~~\begin{tikzpicture}
\draw (-1, 0)  node[rectangle, draw, inner sep=5pt](c11){$H$};
\draw (5, 0)  node[rectangle, draw, inner sep=5pt](c12){$H$};

\draw (1, 0) node[circle, draw, thick](t1){};
\draw (1, -1) node[circle, draw, thick, fill=black, inner sep=1pt](c1){};
\draw[thick] (c1) -- (1, 0.2);

\draw (3, 0) node[circle, draw, thick](t1){};
\draw (3, -1) node[circle, draw, thick, fill=black, inner sep=1pt](c1){};
\draw[thick] (c1) -- (3, 0.2);

\draw (0, 0) node[circle, draw, thick](t2){};
\draw (0, -2) node[circle, draw, thick, fill=black, inner sep=1pt](c2){};
\draw[thick] (c2) -- (0, 0.2);

\draw (4, 0) node[circle, draw, thick](t2){};
\draw (4, -2) node[circle, draw, thick, fill=black, inner sep=1pt](c2){};
\draw[thick] (c2) -- (4, 0.2);

\draw (2, 0)  node[rectangle, draw, inner sep=2pt](g1){\small $R_Z(\theta)$};

\draw (-1, -1)  node[rectangle, draw, inner sep=5pt](c21){\tiny$\sqrt{X}$};
\draw (5, -1)  node[rectangle, draw, inner sep=5pt](c22){\tiny$\sqrt{X}^\dagger$};

\draw (c11) -- (g1) -- (c12);
\draw (c21) -- (c22);

\draw (6, 0)  -- (c12);
\draw (6, -1)  -- (c22);
\draw (6, -2)  -- (-2, -2);
\draw (6, -3)  --(-2, -3);

\draw (-2, 0)  -- (c11);
\draw (-2, -1)  -- (c21);

\end{tikzpicture}  
    \caption{Reduction of a Pauli operator/rotation. (a) the initial Pauli operator. (b) after conjugation via local Cliffords, our operator is diagonal. (c) after conjugation with the appropriate CNOT gates, our operator is localized on a single qubit (here, the first qubit). (d) the final quantum circuit implementing $R_{XYZI}(\theta)$.}
    \label{fig:pauli_reduction}
\end{figure}

In that setting we will consider that $\mathcal{G}_{in}$ contains only Clifford gates and arbitrary Pauli rotations, $R_P$ for $P \in \mathcal{P}$. $\mathcal{G}_{out}$ will contain $CNOT, H, R_X(\pi/2)$, and arbitrary $R_Z$ rotations, the CNOTs being restricted to some interaction graph $G$.

In order to use our meta-heuristic, we need to specify our full data structure $\mathcal{D} = \langle \mathcal{T}, \interp{.}, S, u, e \rangle$:
\begin{itemize}
    \item $\mathcal{T}$ is the set of Clifford operators, or, to be precise, of tableaux representing Clifford operators,
    \item $\interp{.}$ is the standard tableau interpretation,
    \item $S$ is the set of Clifford gates,
    \item $u$ is the update of a tableau using a Clifford gate by left composition:
    $$ u(T, g) = \tilde{g} \cdot T$$
\end{itemize}

Our extraction function $e$ acts as follows. Upon encountering a non-Clifford Pauli rotation $R_P(\theta)$:
\begin{itemize}
    \item[(i)] Compute a Pauli operator $P'$ and a phase $s=\pm 1$ such that $s\cdot P' = T^\dagger PT$
    \item[(ii)] For each qbit $i$ in the support of $P'$, if $P'[i]=Y$ then perform a $R_X(\pi/2)$ gate on $i$, and if $P'[i]=X$, perform a $H$ on $i$. This produces a Clifford circuit $c$, comprising only local gates. E.g $P' = IXYZI$, we produce a circuit $c=H_1\concat R_X(\pi/2)_2$.
    \item[(iii)] Pick a target qubit $q$ in the support of $P'$, and perform algorithm \ref{alg:fan_in} in order to generate a fan-in CNOT from all qubits in the support of $P'$ to $q$, thus updating the Clifford sub-circuit $c$
    \item [(iv)] Update $T$ by right composition with $\tilde{c}^\dagger$:
    $$ T' \gets  T \cdot \tilde{c}^\dagger$$
    \item[(v)] Return the updated table $T'$ and sub-circuit $c\concat R_Z(s\cdot \theta)_q$
\end{itemize}

The following proposition about $e$ holds:
\begin{proposition}
Let $T$ be a tableau and $R_P(\theta)$ be a Pauli rotation. If $T', c = e(T, R_P(\theta))$, then $R_P(\theta).\interp{T} = \interp{T'}.\tilde{c}$
\end{proposition}
\begin{proof}

By construction, we have that:
$$c = c_{prep}\concat R_{Z}(s\cdot \theta)_{q}$$
with $c_{prep}$ and $q$ such that:
$$ c_{prep}\concat R_{Z}(s\cdot \theta)_{q}\concat c_{prep}^\dagger = R_{P'}(s\cdot\theta) $$
where $s\cdot P' = T^\dagger PT $, and $c_{prep}$ is a Clifford circuit.
This implies that $\tilde{c} = \tilde{c_{prep}}\cdot R_{P'}(s\cdot\theta)$

To be precise, $c_{prep}$ holds the local basis changes and CNOT cascade necessary to the implementation of $R_{P'}(s \cdot \theta)$, plus some stray Clifford operators that might have happened during the ``dirty'' fan-in (corresponding to the dashed box in the example circuit below).
\begin{center}
\hspace{2cm}\Qcircuit @C=0.9em @R=0.7em {
&\qw & \gate{C_0} & \qw        & \qw    &  \qw       & \targ     & \gate{R_Z(s\cdot \theta)}  & \targ     & \qw       & \qw & \qw      & \gate{C_0^\dagger} & \multigate{3}{\ \ \ \interp{T}\ \ \ }\\
&\qw &  \qw       & \qw        &\ctrl{1}&\targ       & \ctrl{-1} & \qw                        & \ctrl{-1} & \targ     & \ctrl{1} & \qw      & \qw & \ghost{\ \ \ \interp{T}\ \ \ }\\
&\qw & \gate{C_2} & \targ      & \targ  & \ctrl{-1}  &  \qw      & \qw                        & \qw       & \ctrl{-1} & \targ & \targ    & \gate{C_2^\dagger} & \ghost{\ \ \ \interp{T}\ \ \ }\\
&\qw & \gate{C_3} & \ctrl{-1}  & \qw    & \qw        & \qw       &  \qw                       & \qw       & \qw       & \qw & \ctrl{-1}& \gate{C_3^\dagger} & \ghost{\ \ \ \interp{T}\ \ \ }  \gategroup{1}{1}{4}{7}{1.7em}{--}
}
\end{center}
Hence, we have:
\begin{align*}
    R_P(\theta)\cdot\interp{T} & = \interp{T}\cdot\interp{T}^\dagger \cdot R_{P}(\theta)\cdot \interp{T}\\
                       & = \interp{T}\cdot\left(\interp{T}^\dagger \cdot R_{P}(\theta)\cdot \interp{T}\right)\\
                       & = \interp{T}\cdot R_{T^\dagger PT}(\theta)\\
                       & = \interp{T}\cdot R_{P'}(s\cdot \theta) &&\text{where $s.P' = T^\dagger P T$}\\
                       & = \interp{T}\cdot \tilde{c_{prep}}^\dagger\cdot R_{Z_q}(s\cdot\theta)\cdot\tilde{c_{prep}}\\
                       & = \interp{T\cdot \tilde{c_{prep}}^\dagger} \cdot \tilde{c} \\
                       & = \interp{T'}\cdot \tilde{c} \\
\end{align*}

\end{proof}

Enventually, our final output circuit will always have shape:
$$C_{out} = C \prod_{i} R_{Z_{q_i}}(\theta_i) F_i L_i $$
where $C$ is some Clifford operator, $R_{Z_{q_i}}(\theta_i)$ are non-Clifford local $Z$ rotations, $F_i$ are architecture compliant fan-in CNOTs as described by algorithm \ref{alg:fan_in}, $L_i$ are local Clifford circuits, and $q_i$ are the target qubits used in the Pauli rotation reductions.
Appendix \ref{app:example_clifford} gives a step by step overview of the execution of this algorithm on a simple circuit.

\subsection{Further optimizations}

\noindent{\bf Recursive search of finite depth.} Notice that, once again, we have some freedom of choice when picking the qubit that will receive the $R_Z$ rotation. In practice, we perform a recursive search of finite depth for the next $w$ rotations to synthesize and pick the host qubit that leads to the least overhead. The branching is very similar to the one described in \ref{ssec:ext_cnots}. After computing the Steiner tree with terminal vertices the support of the rotation we are currently synthesizing, one can choose any terminal vertex to be the target of our fan-in. Once again we refer to Appendix \ref{sec:rec_search} for more details.
The overall worst case complexity is the same as the CNOT case. Indeed, the complexity is dominated by the recursive exploration of a search tree where each vertex exploration requires the generation of a Steiner tree of the architecture graph.

\medskip

\noindent{\bf Dealing with the final Clifford operator.} Once again, we are left with a possibly non-trivial final Clifford operator $C$. As stated in the previous section, if one has to compile several pieces of circuits in sequence, one can always initialize the Clifford operator of the next compilation round using $C$.
In the general case where we are done compiling and need to effectively deal with this operator, we can almost always avoid having to synthesize the full operator $C$. Section \ref{sec:opt} describes how to do so when sampling an observable or sampling bit-strings in the computational basis.

\medskip

\noindent{\bf Rotation merging.} As mentioned in subsection \ref{subsec:clifford}, any quantum circuit can be reformulated as a sequence of Pauli rotations with non-Clifford angles (i.e. angles $\neq k\frac{\pi}{2}$). That is:
$$C\prod_i R_{P_i}(\theta_i)$$
where $R_{P_i}(\theta_i)$ are Pauli rotations and $C$ is a final Clifford operator. Moreover, this form can be efficiently computed by pulling all the Clifford gates at the end of the circuit.
Once such a product is obtained, one can try to merge rotations with identical axis. This can also be done efficiently by considering each rotation one by one and checking if it can be commuted and merged with a rotation with an identical axis. This optimization corresponds exactly the the T-count reduction algorithm presented in \cite{zhang2019optimizing}. This routine is described in Algorithm \ref{alg:rotation_merging}. Notice that this is not the only way to produce a final ordering of the rotations. In particular when we insert the un-merged rotation in list $L$ (line 14), one would make a different choice and insert sooner in the list. By inserting it at the end of the list, we might block some other merges by preventing the next rotations to commute past it. In order to keep this optimization lightweight and reproducible, we keep things simple and insert the rotation at the end of the list. 

This optimization has several consequences. First, by merging rotations, we reduce the number of calls to the partial synthesis routine. Moreover, by merging rotations, one might end up with a rotation with Clifford angle. Such a rotation can then be pulled and the end of the circuit, effectively removing it from the sequence of rotations to synthesize. This optimization is a key feature when dealing with Clifford + T circuits where this type of situation occurs regularly.
This pre-processing has a worst case time complexity of $O(m^2n)$ where $m$ is the number of non-Clifford Pauli rotations and $n$ is the number of qubits.

\begin{algorithm}[h]
    \caption{Rotation merging}\label{alg:rotation_merging}
    \SetKwInOut{Input}{input}
    \SetKwInOut{Output}{output}
    \SetKwComment{Comment}{//}{}
    \Input{a rotation sequence $S$}
    \Output{a rotation sequence $L$}

     $L \gets [~]$\;
    \For{$R_P(\theta)$ in $S$}{
        \For{$R_{P'}(\theta')$ in $\textrm{reversed}(L)$}{
            \If{$P = P'$}{
                 $\theta' \gets \theta' + \theta$\;
                 break\;
            }
            \If{$P$ and $P'$ do not commute}{
                 break\;
            }
        }
        \If{$P$ was not inserted}{
             $L \gets L::R_P(\theta)$\;
        }
    }
    \Return $L$
\end{algorithm}

\medskip

\noindent{\bf Rotation reordering.} Another optimization that can easily be computed is the reordering of consecutive commuting rotations. Given a sequence of Pauli rotations $\prod R_{P_i}(\theta_i)$, one can rewrite it as $\prod_G \prod_{i\in G} R_{P_i}(\theta_i)$ where $G$ are groups of commuting rotations. Notice that in this expression, while the first product is ordered, the second is not. This gives us a leverage for optimization. In practice, we use a greedy approach consisting in synthesizing the less costly rotation first. That is, we compute all the Steiner trees necessary to implement all the rotations in a given group and start with the rotation that requires the smallest tree.
Groups of commuting rotations are computed greedily using Algorithm \ref{alg:group_rotations}. Notice that this is not the only way to produce such a sequence. In practice, trying harder to form larger groups of commuting rotation did not seem to improve the benchmark results, hence the rather simple greedy heuristic. 
This pre-processing has a worst case time complexity of $O(m^2n)$ where $m$ is the number of non-Clifford Pauli rotations and $n$ is the number of qubits.

\begin{algorithm}[H]
    \caption{Rotation grouping}\label{alg:group_rotations}
    \SetKwInOut{Input}{input}
    \SetKwInOut{Output}{output}
    \SetKwComment{Comment}{//}{}
    \Input{a rotation sequence $S$}
    \Output{a sequence of rotation groups $L$}
     $L \gets [~]$\;
     $G \gets \{\}$\;
    \For{$R_P(\theta)$ in $S$}{
        \If{$R_P$ commutes with all rotations in $G$}{
             $G.insert(R_P(\theta))$\;
             {\bf continue}\;
        }
        $L \gets L::G$\;
        $G \gets \{R_P(\theta)\}$\;
    }
    $L \gets L::G$\;
    \Return $L$\;
\end{algorithm}

\section{Benchmarks} \label{sec:benchmarks}

In order to benchmark our method we picked three representative architectures: Rigetti's Aspen chip (16 qubits), IBM's Melbourne chip (14 qubits), and a fictive all-to-all (14 qubits) architecture. The idea being that Melbourne's connectivity is close to a grid, whereas Aspen's connectivity contains longer cycles and has a less regular structure. The all-to-all architecture is here to act as a baseline in the benchmarks. The connectivity graphs are described in Figure \ref{fig:architectures}.

\begin{figure}[h]
    \centering
    (a)~~\begin{tikzpicture}
        \draw (0,0) node[circle, draw, inner sep=2pt](n0){$0$};
        \draw (1,0) node[circle, draw, inner sep=2pt](n1){$1$};
        \draw (2,0) node[circle, draw, inner sep=2pt](n2){$2$};
        \draw (3,0) node[circle, draw, inner sep=2pt](n3){$3$};
        \draw (4,0) node[circle, draw, inner sep=2pt](n4){$4$};
        \draw (5,0) node[circle, draw, inner sep=2pt](n5){$5$};
        \draw (6,0) node[circle, draw, inner sep=2pt](n6){$6$};
        
        \draw (1,-1) node[circle, draw, inner sep=1pt](n13){$13$};
        \draw (2,-1) node[circle, draw, inner sep=1pt](n12){$12$};
        \draw (3,-1) node[circle, draw, inner sep=1pt](n11){$11$};
        \draw (4,-1) node[circle, draw, inner sep=1pt](n10){$10$};
        \draw (5,-1) node[circle, draw, inner sep=2pt](n9){$9$};
        \draw (6,-1) node[circle, draw, inner sep=2pt](n8){$8$};
        \draw (7,-1) node[circle, draw, inner sep=2pt](n7){$7$};
        \draw (n0) -- (n1) -- (n2)-- (n3)-- (n4)-- (n5)-- (n6)-- (n8)-- (n9)-- (n10)-- (n11)-- (n12)-- (n13)-- (n1);
        \draw (n2) -- (n12);
        \draw (n3) -- (n11);
        \draw (n4) -- (n10);
        \draw (n5) -- (n9);
        \draw (n8) -- (n7);
    \end{tikzpicture}
    ~~~~(b)~~\begin{tikzpicture}
        \draw (-22.5 - 45:1cm) node[circle, draw, inner sep=2pt](n0){$0$};
        \draw (-22.5 :1cm) node[circle, draw, inner 
        sep=2pt](n1){$1$};
        \draw (-22.5 + 45:1cm) node[circle, draw, inner sep=2pt](n2){$2$};
        \draw (-22.5 + 90:1cm) node[circle, draw, inner sep=2pt](n3){$3$};
        \draw (-22.5 + 45 + 90:1cm) node[circle, draw, inner sep=2pt](n4){$4$};
        \draw (-22.5 + 180:1cm) node[circle, draw, inner sep=2pt](n5){$5$};
        \draw (-22.5 + 180 + 45:1cm) node[circle, draw, inner sep=2pt](n6){$6$};
        \draw (-22.5 + 270:1cm) node[circle, draw, inner sep=2pt](n7){$7$};
        
        \draw (3, 0) + (-22.5 - 45:1cm) node[circle, draw, inner sep=2pt](n8){$8$};
        \draw (3, 0) + (-22.5 :1cm) node[circle, draw, inner 
        sep=2pt](n9){$9$};
        \draw (3, 0) + (-22.5 + 45:1cm) node[circle, draw, inner sep=1pt](n10){$10$};
        \draw (3, 0) + (-22.5 + 90:1cm) node[circle, draw, inner sep=1pt](n11){$11$};
        \draw (3, 0) + (-22.5 + 45 + 90:1cm) node[circle, draw, inner sep=1pt](n12){$12$};
        \draw (3, 0) + (-22.5 + 180:1cm) node[circle, draw, inner sep=1pt](n13){$13$};
        \draw (3, 0) + (-22.5 + 180 + 45:1cm) node[circle, draw, inner sep=1pt](n14){$14$};
        \draw (3, 0) + (-22.5 + 270:1cm) node[circle, draw, inner sep=1pt](n15){$15$};
        \draw (n0) -- (n1)-- (n2)-- (n3)-- (n4)-- (n5)-- (n6)-- (n7) -- (n0);
        \draw (n8) -- (n9)-- (n10)-- (n11)-- (n12)-- (n13)-- (n14)-- (n15) -- (n8);
        \draw (n2) -- (n13);
        \draw (n1) -- (n14);
    \end{tikzpicture}
    \caption{(a) IBM's Melbourne and (b) Rigetti's Aspen connectivity graph}
    \label{fig:architectures}
\end{figure}
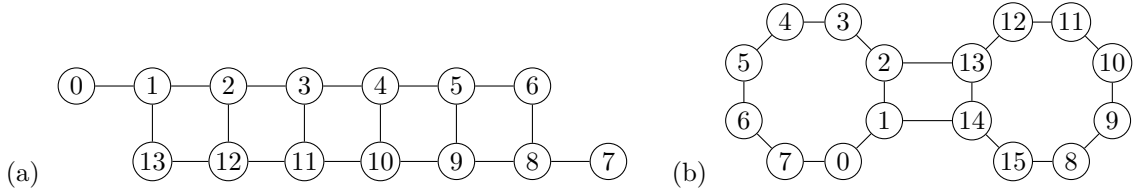

We conducted two sets of benchmarks whose results are reported in the Subsections \ref{ssec:full_benchs} 
and \ref{ssec:scalable_benchs}.

\subsection{Benchmarks content}\label{ssec:full_benchs}

We benchmarked\footnote{Benchmark data can be found at \url{https://github.com/smartiel/lazy_synthesis_benchmarks} together with some compiled circuits.} eight different algorithms: 
\begin{itemize}
    \item Hirata \emph{et al} SWAP insertion algorithm (generalized to arbitrary connectivity, search depth of $4$), denoted \emph{swap} in the various benchmarks (algorithm of Section \ref{sec:swaps}),
    \item lazy synthesis using linear boolean operators (depth of $3$), denoted \emph{linear} in the benchmarks (algorithm of Section \ref{sec:cnots}),
    \item lazy synthesis using Clifford operators (depth of $3$), denoted \emph{clifford} in the benchmarks  (algorithm of Section \ref{sec:cliffords}). This one comes in $4$ variants: with and without rotation reordering ($\star$ marker) and with and without rotation merging ($\dagger$ marker).
    \item Cowtan et al. swap insertion algorithm (see \cite{cowtan2019qubit}), denoted \emph{cowtan} in the benchmarks,
    \item Staq compiler \cite{Amy_2020} without initial mapping optimization and with rotation merging to reduce T-count (i.e. exactly the rotation merging pre-processing described in algorithm \ref{alg:rotation_merging}), denoted \emph{staq} in the benchmarks,
\end{itemize} 

We ran these algorithms on various sets of quantum circuits:
\begin{itemize}
    \item A collection of standard circuits taken from \cite{amy2018controlled} that fit on $14$ qubits. Circuits are simply pre-processed by replacing Toffoli gates by a standard CNOT + T decomposition. Tables \ref{fig:table_arithmetic_melbourne}, \ref{fig:table_arithmetic_aspen}, \ref{fig:table_arithmetic_ata} provides the relative CNOT overhead for the three hardware models.
    \item A set of random QAOA instances of MAX-k-LIN-2 (depth 1). These circuits are basically phase polynomials where parities have uniform Hamming weights equal to $k$. The input circuits are generated using a naive strategy and produce a large amount of CNOTs (each parity is implemented via two cascades of CNOTs and a $R_Z$ rotation). Their Clifford density roughly grows as $\frac{2(k-1)}{2(k-1) + 1}$ (neglecting the final layer of non-Clifford $X$ rotations and the initial Walsh-Hadamard transform).
    \item A set of random products of arbitrary Pauli rotations. These present roughly the same statistical features as standard quantum chemistry/material Ansätze. These circuits usually exhibit quite large Clifford densities ($>0.9$).
\end{itemize}

 For the last two benchmarks (QAOA and Pauli rotations), the rotation merging optimization is pointless since we make sure that all rotations have distinct axis in the input circuit. Thus, for these benchmarks we only display the performances of the $\star$ version of the \emph{clifford} algorithm.
 
 \smallskip
 
\noindent{\bf Why no other SWAP insertion algorithms?} We also tried to include other SWAP insertion algorithms (namely SABRE \cite{sabre}, and $A^*$ based approach \cite{bka}), but both of these methods performed systematically worse than Hirata \emph{et al} approach generalized to arbitrary connectivity (our algorithm \emph{swap} described in section \ref{sec:swaps}). Moreover, the execution time of the $A^*$ approach can sometimes become prohibitive, which makes it unpractical even for small architectures. Notice also that \cite{cowtan2019qubit} shares a lot of similarities with SABRE.

\begin{table}[h]
    \centering
    \caption{Compilation of a collection of standard circuits for Melbourne architecture.}
    \label{fig:table_arithmetic_melbourne}
    \begin{center}
    \scalebox{0.85}{
\begin{tabular}{|c|c|c|c|c|c|c|c|c|c|}
\hline
circuit & init & swap & linear & clifford & clifford$\star$ & clifford$\dagger$ & clifford$\star\dagger$ & cowtan & staq\\
\hline
tof\_3 & 18 & 116.7\% & 150.0\% & 138.9\% & 77.8\% & 127.8\% & \textcolor{green}{72.2\%} & 133.3\% & 77.8\%\\
barenco\_tof\_3 & 24 & 75.0\% & 66.7\% & 66.7\% & 50.0\% & 25.0\% & \textcolor{green}{-4.2\%} & 100.0\% & 45.8\%\\
mod5\_4 & 28 & 117.9\% & 60.7\% & 25.0\% & -3.6\% & 0.0\% & \textcolor{green}{-21.4\%} & 171.4\% & 117.9\%\\
tof\_4 & 30 & 110.0\% & 150.0\% & 120.0\% & 103.3\% & 116.7\% & \textcolor{green}{83.3\%} & 160.0\% & 100.0\%\\
tof\_5 & 42 & 135.7\% & 276.2\% & 226.2\% & 214.3\% & 157.1\% & 109.5\% & 150.0\% & \textcolor{green}{54.8\%}\\
qft\_4 & 46 & 176.1\% & 60.9\% & 28.3\% & 19.6\% & \textcolor{green}{-23.9\%} & -19.6\% & 117.4\% & 56.5\%\\
barenco\_tof\_4 & 48 & 112.5\% & 170.8\% & 87.5\% & 87.5\% & 12.5\% & \textcolor{green}{0.0\%} & 150.0\% & 60.4\%\\
mod\_mult\_55 & 48 & 337.5\% & 345.8\% & 220.8\% & 181.2\% & 172.9\% & \textcolor{green}{168.8\%} & 193.8\% & 306.2\%\\
vbe\_adder\_3 & 70 & 107.1\% & 60.0\% & 38.6\% & 11.4\% & \textcolor{green}{-32.9\%} & -17.1\% & 120.0\% & 135.7\%\\
barenco\_tof\_5 & 72 & 112.5\% & 245.8\% & 119.4\% & 127.8\% & 41.7\% & \textcolor{green}{20.8\%} & 137.5\% & 59.7\%\\
rc\_adder\_6 & 93 & 180.6\% & 76.3\% & 31.2\% & 31.2\% & -7.5\% & \textcolor{green}{-10.8\%} & 112.9\% & 221.5\%\\
gf2\^~4\_mult & 99 & 184.8\% & 278.8\% & 205.1\% & 93.9\% & 180.8\% & \textcolor{green}{84.8\%} & 197.0\% & 381.8\%\\
mod\_red\_21 & 105 & 165.7\% & 204.8\% & 116.2\% & 105.7\% & 79.0\% & \textcolor{green}{58.1\%} & 171.4\% & 210.5\%\\
hwb6 & 116 & 196.6\% & 169.0\% & 91.4\% & 64.7\% & 67.2\% & \textcolor{green}{52.6\%} & 152.6\% & 205.2\%\\
grover\_5 & 288 & 116.7\% & 210.4\% & 245.1\% & 166.3\% & 194.4\% & \textcolor{green}{91.7\%} & 121.9\% & 92.0\%\\
hwb8 & 7129 & 224.2\% & 168.5\% & 169.7\% & 156.6\% & 134.4\% & \textcolor{green}{114.1\%} & 183.6\% & 280.5\%\\
\hline
\end{tabular}
    }
    \end{center}
    
\end{table}

\begin{table}[h]
    \centering
    \caption{Compilation of a collection of standard circuits for Aspen architecture.}
    \label{fig:table_arithmetic_aspen}
    \begin{center}
    \scalebox{0.85}{
    \begin{tabular}{|c|c|c|c|c|c|c|c|c|c|}
\hline
circuit & init & swap & linear & clifford & clifford$\star$ & clifford$\dagger$ & clifford$\star\dagger$ & cowtan & staq\\
\hline
tof\_3 & 18 & 116.7\% & 161.1\% & 127.8\% & 77.8\% & \textcolor{green}{72.2\%} & \textcolor{green}{72.2\%} & 150.0\% & 83.3\%\\
barenco\_tof\_3 & 24 & 75.0\% & 66.7\% & 66.7\% & 50.0\% & 25.0\% & \textcolor{green}{-4.2\%} & 137.5\% & 58.3\%\\
mod5\_4 & 28 & 117.9\% & 60.7\% & 25.0\% & -3.6\% & 0.0\% & \textcolor{green}{-21.4\%} & 160.7\% & 117.9\%\\
tof\_4 & 30 & 110.0\% & 223.3\% & 130.0\% & 103.3\% & 106.7\% & \textcolor{green}{70.0\%} & 120.0\% & 76.7\%\\
tof\_5 & 42 & 164.3\% & 238.1\% & 126.2\% & 166.7\% & \textcolor{green}{111.9\%} & 164.3\% & 242.9\% & 216.7\%\\
qft\_4 & 46 & 176.1\% & 60.9\% & 28.3\% & 19.6\% & \textcolor{green}{-23.9\%} & -19.6\% & 123.9\% & 41.3\%\\
barenco\_tof\_4 & 48 & 112.5\% & 179.2\% & 106.2\% & 56.2\% & 12.5\% & \textcolor{green}{0.0\%} & 118.8\% & 85.4\%\\
mod\_mult\_55 & 48 & 181.2\% & 300.0\% & 133.3\% & 120.8\% & 104.2\% & \textcolor{green}{100.0\%} & 187.5\% & 210.4\%\\
vbe\_adder\_3 & 70 & 145.7\% & 94.3\% & 81.4\% & 77.1\% & 27.1\% & \textcolor{green}{22.9\%} & 171.4\% & 35.7\%\\
barenco\_tof\_5 & 72 & 125.0\% & 220.8\% & 229.2\% & 166.7\% & \textcolor{green}{55.6\%} & 56.9\% & 133.3\% & 273.6\%\\
rc\_adder\_6 & 93 & 190.3\% & 221.5\% & 107.5\% & \textcolor{green}{72.0\%} & 82.8\% & 82.8\% & 254.8\% & 255.9\%\\
gf2\^~4\_mult & 99 & 257.6\% & 382.8\% & 312.1\% & \textcolor{green}{151.5\%} & 270.7\% & 191.9\% & 242.4\% & 470.7\%\\
mod\_red\_21 & 105 & 162.9\% & 289.5\% & 196.2\% & 126.7\% & 181.0\% & \textcolor{green}{123.8\%} & 188.6\% & 198.1\%\\
hwb6 & 116 & 178.4\% & 160.3\% & 81.0\% & 53.4\% & 59.5\% & \textcolor{green}{41.4\%} & 155.2\% & 261.2\%\\
grover\_5 & 288 & 129.2\% & 209.7\% & 222.2\% & 149.0\% & 170.5\% & \textcolor{green}{97.9\%} & 146.9\% & 155.2\%\\
hwb8 & 7129 & 199.0\% & 191.4\% & 180.8\% & 205.6\% & 137.1\% & \textcolor{green}{118.8\%} & 227.2\% & 289.7\%\\
\hline
\end{tabular}}
    \end{center}
    
\end{table}

\begin{table}[h]
    \centering
    \caption{Compilation of a collection of standard circuits for all-to-all architecture.}
    \label{fig:table_arithmetic_ata}
    \begin{center}
    \scalebox{0.85}{\begin{tabular}{|c|c|c|c|c|c|c|c|}
\hline
circuit & init & linear & clifford & clifford$\star$ & clifford$\dagger$ & clifford$\star\dagger$ & staq\\
\hline
tof\_3 & 18 & 0.0\% & -5.6\% & -27.8\% & -22.2\% & \textcolor{green}{-38.9\%} & 0.0\%\\
barenco\_tof\_3 & 24 & 0.0\% & 0.0\% & -29.2\% & -33.3\% & \textcolor{green}{-50.0\%} & -8.3\%\\
mod5\_4 & 28 & 10.7\% & -17.9\% & -42.9\% & -39.3\% & \textcolor{green}{-50.0\%} & -3.6\%\\
tof\_4 & 30 & 0.0\% & 0.0\% & -16.7\% & -13.3\% & \textcolor{green}{-30.0\%} & -3.3\%\\
tof\_5 & 42 & 4.8\% & 0.0\% & -16.7\% & -11.9\% & \textcolor{green}{-28.6\%} & -4.8\%\\
qft\_4 & 46 & -17.4\% & -37.0\% & -47.8\% & -54.3\% & \textcolor{green}{-60.9\%} & -28.3\%\\
barenco\_tof\_4 & 48 & 0.0\% & 4.2\% & -18.8\% & \textcolor{green}{-31.2\%} & \textcolor{green}{-31.2\%} & -14.6\%\\
mod\_mult\_55 & 48 & 4.2\% & -2.1\% & -14.6\% & -18.8\% & \textcolor{green}{-22.9\%} & 16.7\%\\
vbe\_adder\_3 & 70 & -4.3\% & -12.9\% & -40.0\% & -48.6\% & \textcolor{green}{-61.4\%} & -18.6\%\\
barenco\_tof\_5 & 72 & 2.8\% & 12.5\% & -6.9\% & -25.0\% & \textcolor{green}{-37.5\%} & -16.7\%\\
rc\_adder\_6 & 93 & 2.2\% & 10.8\% & 10.8\% & -1.1\% & \textcolor{green}{-12.9\%} & 2.2\%\\
gf2\^~4\_mult & 99 & 43.4\% & 56.6\% & \textcolor{green}{-9.1\%} & 27.3\% & 1.0\% & 62.6\%\\
mod\_red\_21 & 105 & 8.6\% & 53.3\% & 24.8\% & 26.7\% & \textcolor{green}{4.8\%} & 12.4\%\\
hwb6 & 116 & 4.3\% & 13.8\% & -7.8\% & -5.2\% & \textcolor{green}{-18.1\%} & 0.0\%\\
grover\_5 & 288 & 0.0\% & 24.7\% & 39.6\% & 5.6\% & 26.4\% & \textcolor{green}{-4.5\%}\\
hwb8 & 7129 & 0.9\% & 69.3\% & 45.3\% & 36.1\% & 21.6\% & \textcolor{green}{-6.6\%}\\
\hline
\end{tabular}}
    \end{center}
    
\end{table}

\subsection{Discussion}

\noindent{\bf Standard circuits.} The Clifford based approach outperforms almost systematically the other approaches in the case of limited connectivity. Interestingly, our \emph{swap} algorithm outperforms the state-of-the-art \emph{cowtan} method of \cite{cowtan2019qubit} $24$ out of $32$ times. The \emph{staq} method can sometimes outperform our methods by a decent margin. Notice also that the most optimized version of the \emph{clifford} method is not necessarily the best method for some circuits. 

\medskip

\noindent{\bf MAX-k-LIN-2 and random Pauli sequences.} For both of these benchmarks the input circuits have a quite high Clifford density, since they are based on an initial naive implementation of a sequence of Pauli rotations. It is interesting to notice that since the MAX-k-LIN-2 circuits roughly correspond to phase polynomials, the \emph{linear} and \emph{clifford}${\star\dagger}$ approaches have very comparable behaviors. The \emph{clifford}${\star\dagger}$ approach, however, can benefit from the rotation reordering optimization. With this optimization, it becomes overwhelmingly better than the \emph{swap} or \emph{linear} approach. Notice that the first point(s) of the graph (Hamming weight of 2) corresponds exactly to MAX-CUT QAOA circuits which are often taken as standard circuits for NISQ era applications. On these circuits, the \emph{clifford}${\star\dagger}$ presents a two fold improvement compared to the SWAP insertion approaches. Notice that, even though the \emph{linear} approach behaves better than the Staq compiler for limited connectivity, it falls behind for all-to-all connectivity.
In the random Pauli setting, the \emph{clifford} approach, by itself, systematically beats the other algorithms. The rotation reordering optimization does not bring significant improvements compared to the standard \emph{clifford} approach. For this class of circuits, the state-of-the-art methods lags behind both the \emph{linear} and \emph{clifford} approaches when considering limited connectivity.

\begin{figure}[h]
    \centering
    \includegraphics[scale=0.6]{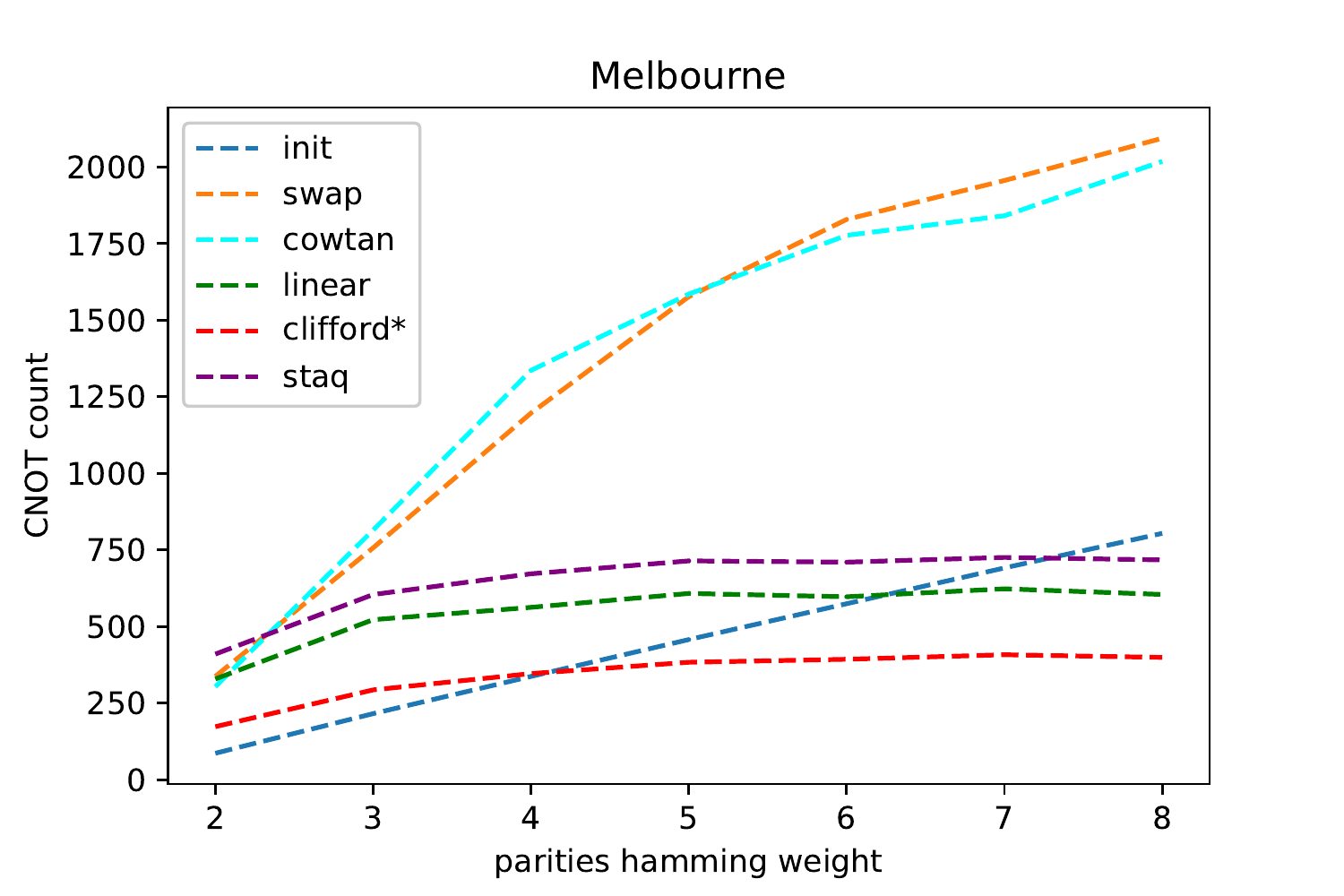}%
    \includegraphics[scale=0.6]{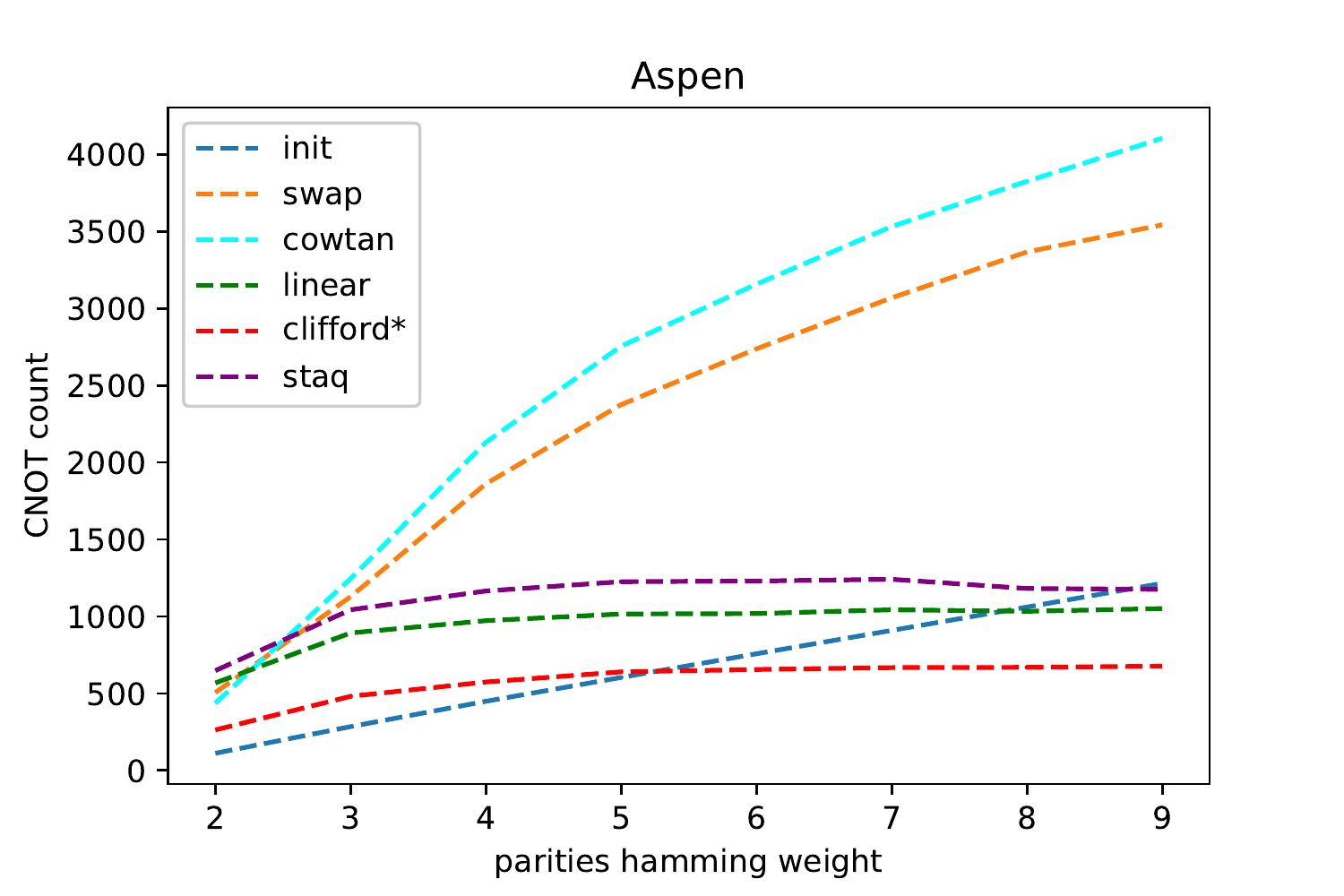}
    \includegraphics[scale=0.6]{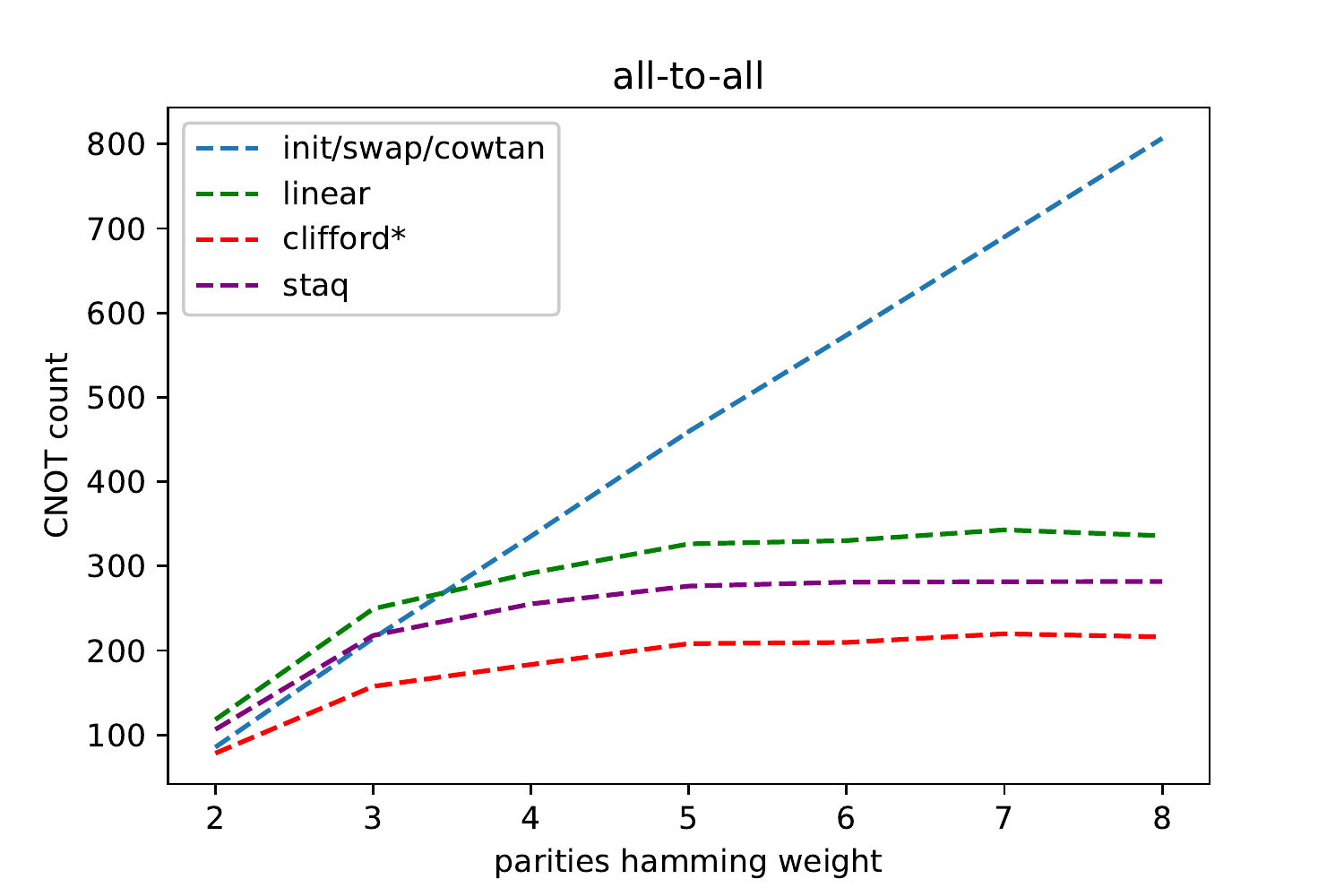}
    \caption{QAOA circuits for MAX-K-LIN-2. Each circuit is generated by picking $n^2$ random parities, without repetitions, of fixed hamming weight (between 2 and $n/2$). Each point is generated using 30 random instances.}
    \label{graph:qaoa}
\end{figure}

\begin{figure}[h]
    \centering
    \includegraphics[scale=0.6]{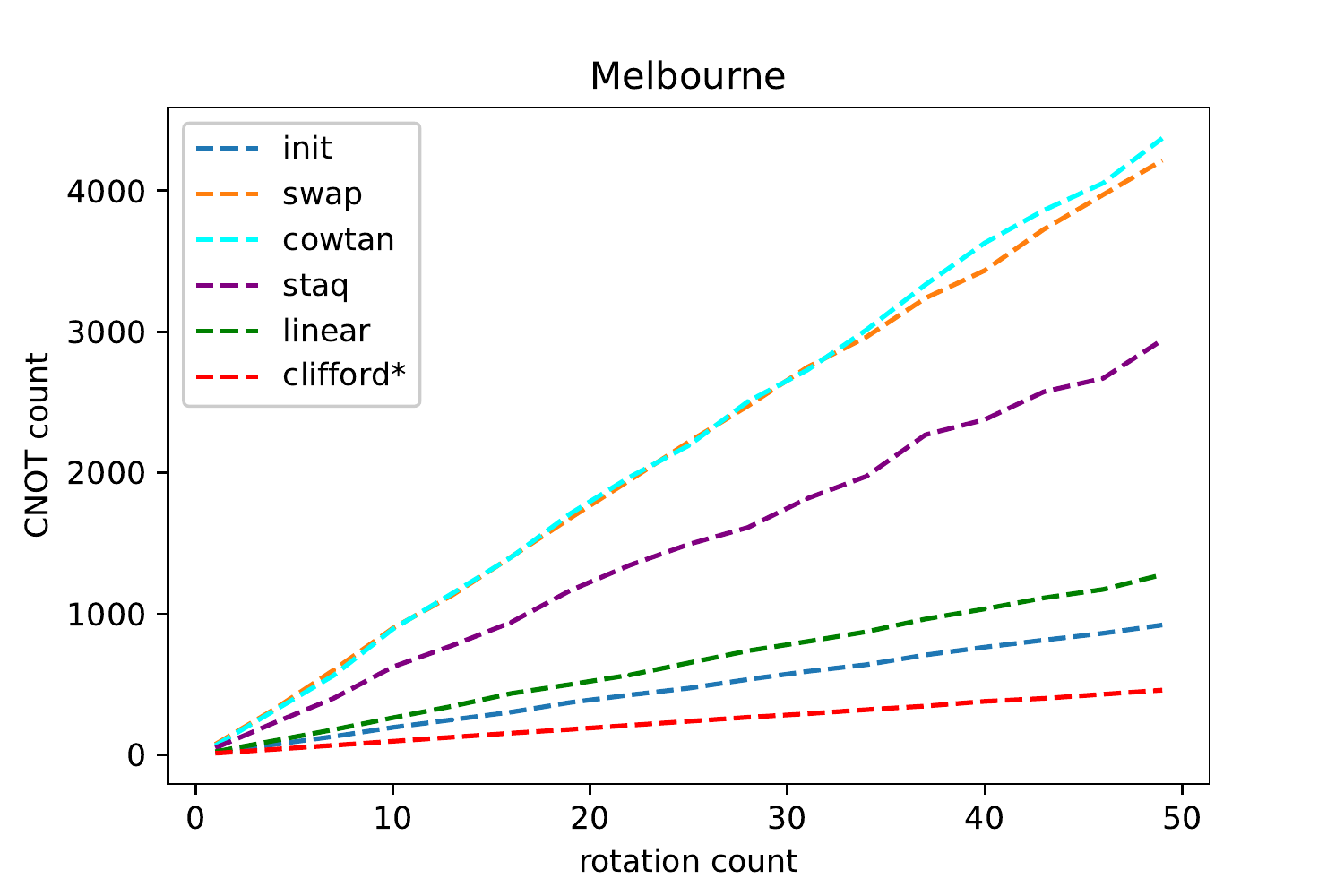}%
    \includegraphics[scale=0.6]{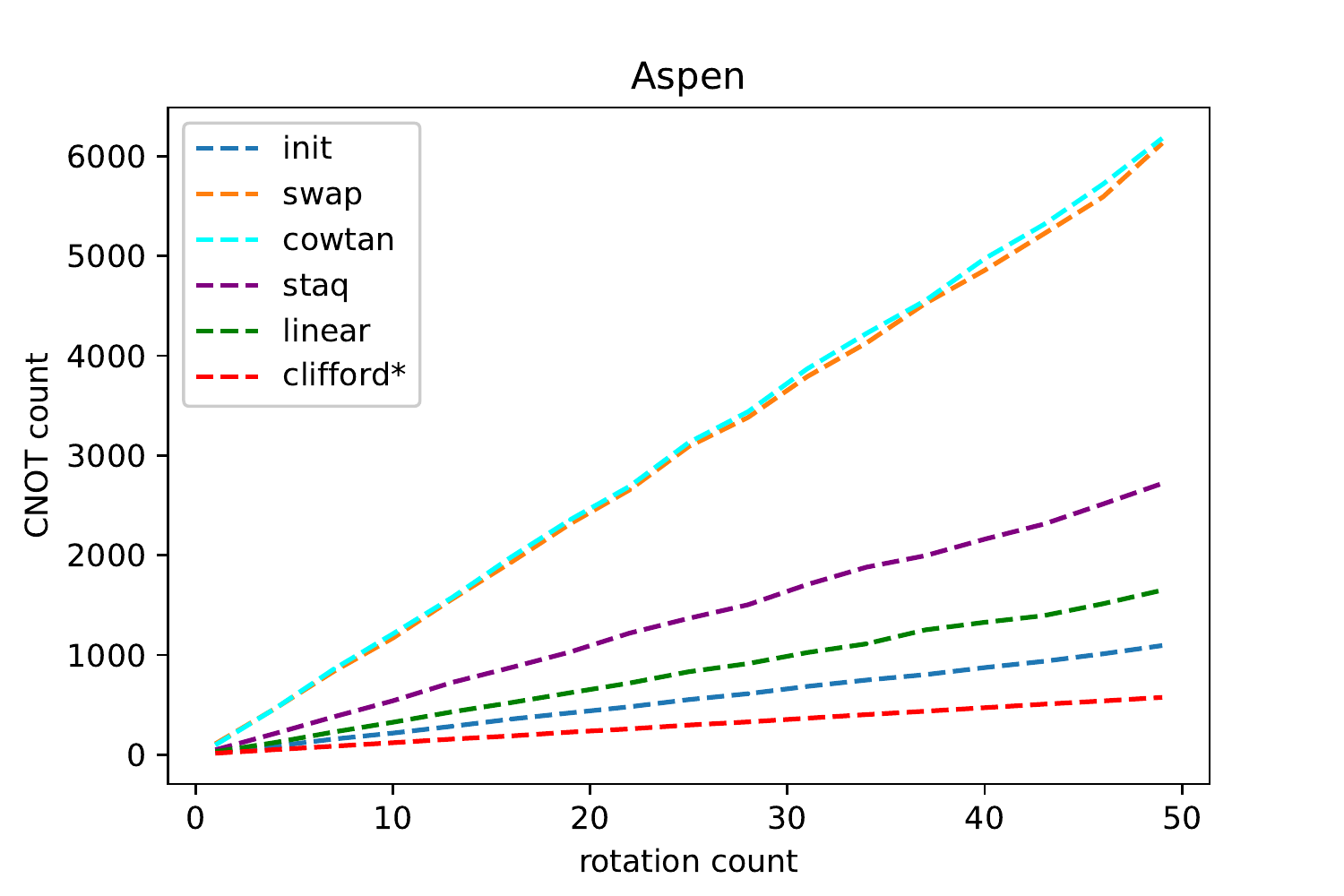}
    \includegraphics[scale=0.6]{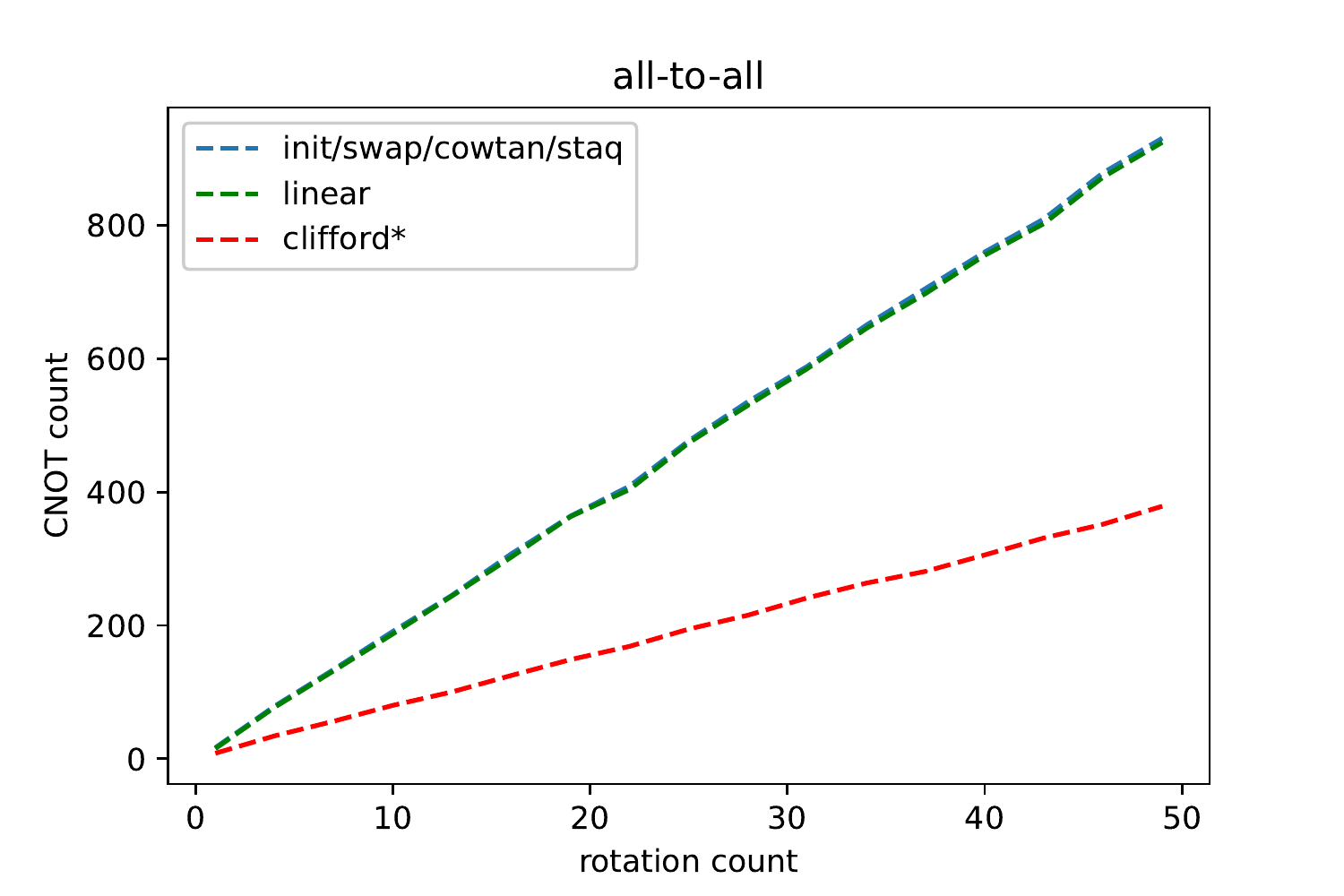}
    
    \caption{Random Pauli rotation sequence. Each circuit is generated by picking at random a collection of Pauli operators (between 1 and 51 random operators, all distinct) and by implementing a naive circuit performing a corresponding sequence of rotations with non-Clifford angles. Each point is generated by picking 10 random operators.}
    \label{graph:pauli}
\end{figure}

\subsection{Scalable benchmarks}\label{ssec:scalable_benchs}
In order to have a better idea of the behavior of the \emph{clifford} method when used with low recursive depth, 
we ran the QAOA and random Pauli rotations benchmarks for the \emph{clifford} method with depths 0 and 1.
These depth values correspond to worst case complexities of $O(n^3m^2)$ and $O(n^4m^2)$ (including all the preprocessings).
Figures \ref{graph:qaoa_low_depth} and \ref{graph:pauli_low_depth} present the benchmarks results.

The results are pretty self explanatory: the \emph{clifford} method remains more than competitive compared to the SWAP insertion 
techniques or the Staq framework. Notice that the depth 0 method is nothing more than a greedy method that iteratively synthesizes the
non-Clifford Pauli rotations in the input circuit. Notice also that going from the greedy version (i.e. depth 0) to the depth 1 
version only marginally improves the (CNOT count) performances of the algorithm.

\begin{figure}[h]
    \centering
    \includegraphics[scale=0.6]{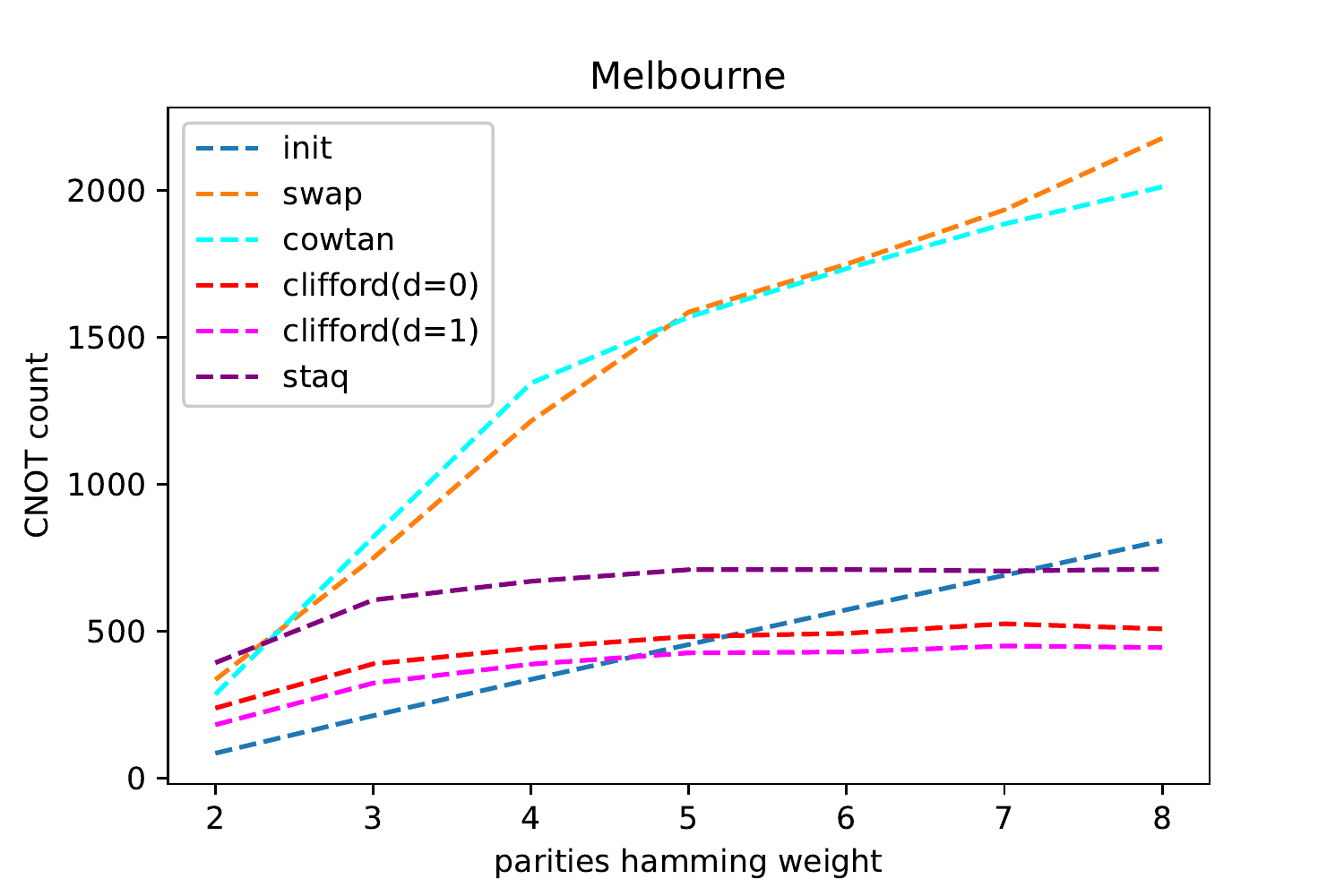}%
    \includegraphics[scale=0.6]{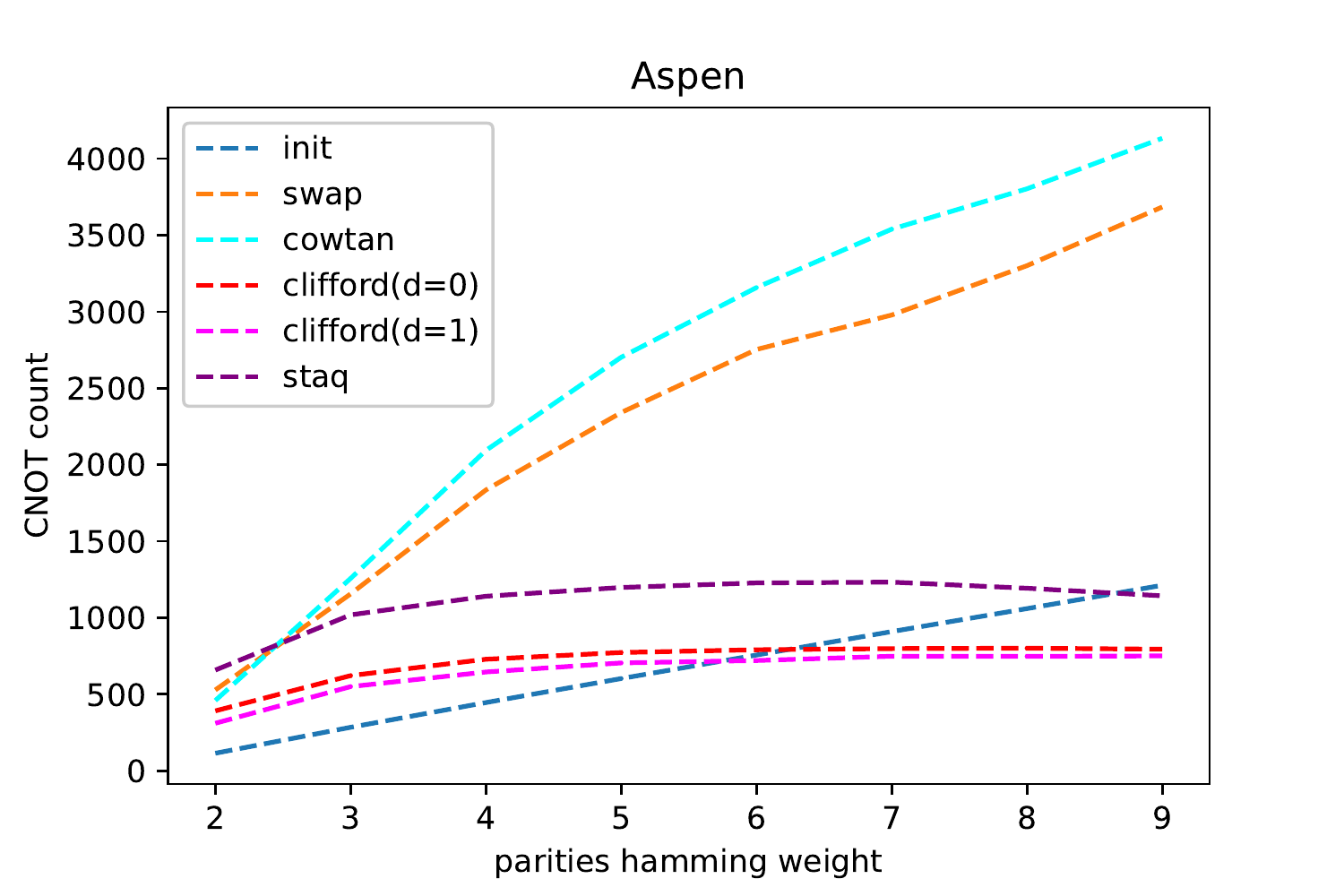}
    \includegraphics[scale=0.6]{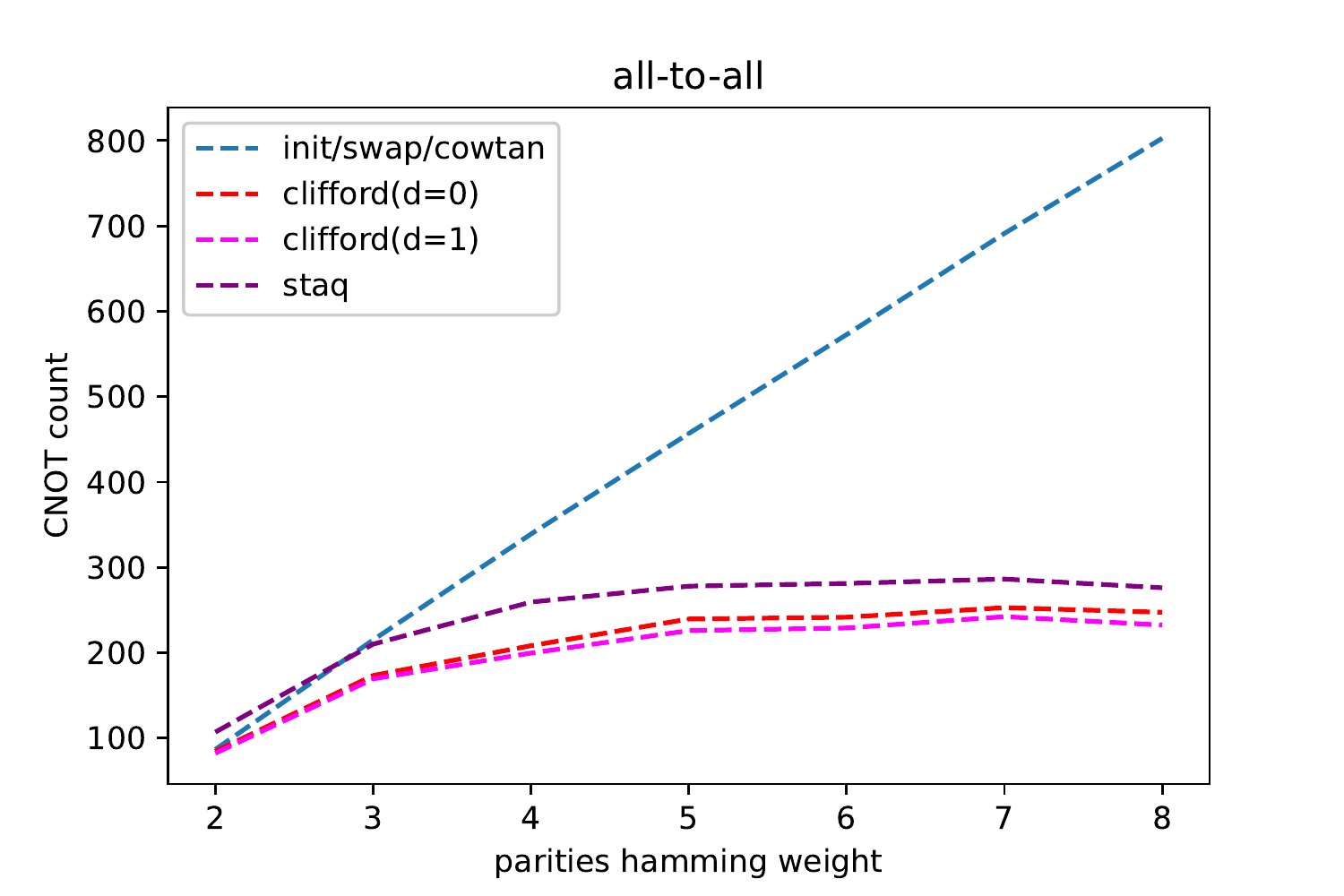}
    \caption{QAOA circuits for MAX-K-LIN-2. Each circuit is generated by picking $n^2$ random parities, without repetitions, of fixed hamming weight (between 2 and $n/2$). Each point is generated using 30 random instances.}
    \label{graph:qaoa_low_depth}
\end{figure}

\begin{figure}[h]
    \centering
    \includegraphics[scale=0.6]{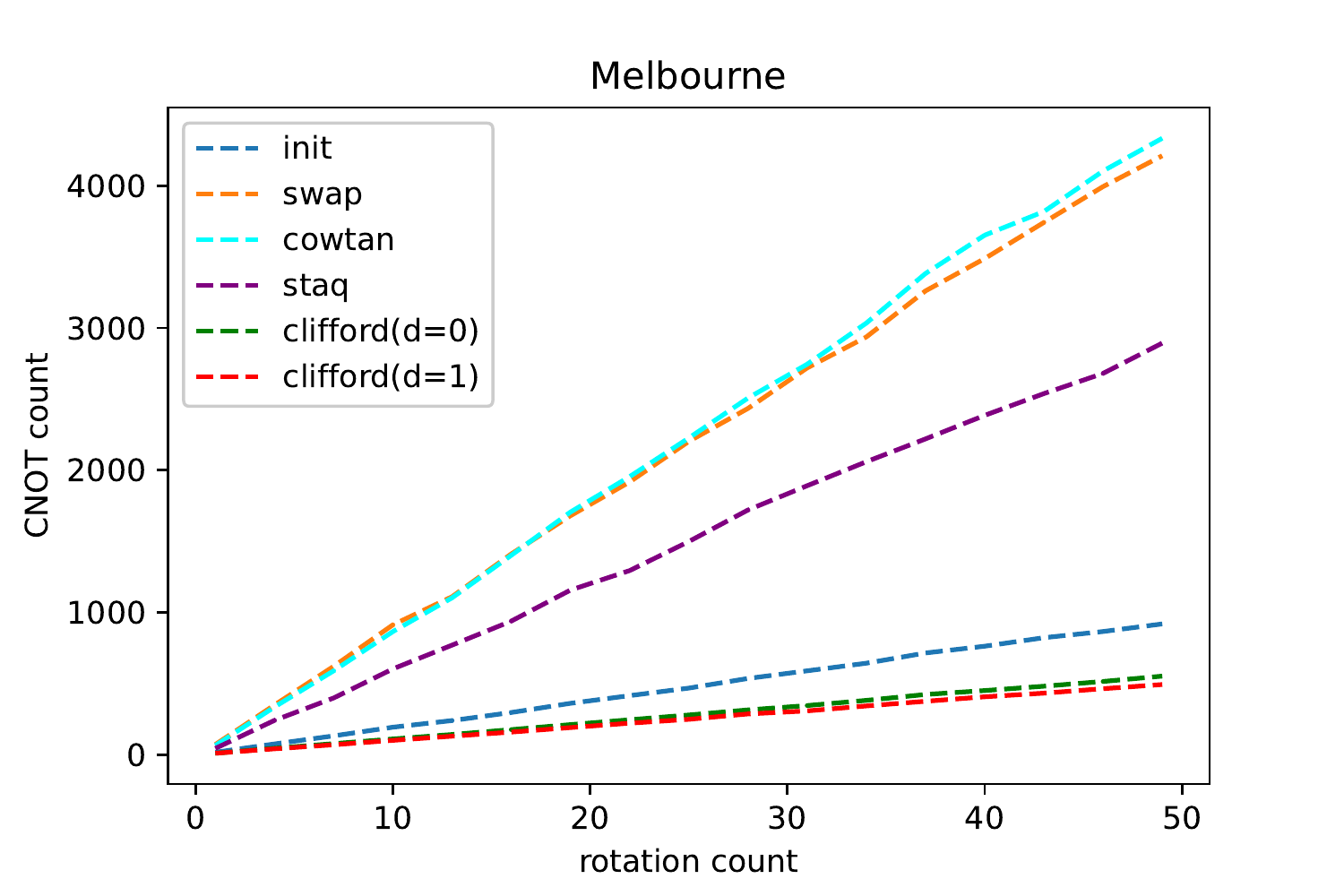}%
    \includegraphics[scale=0.6]{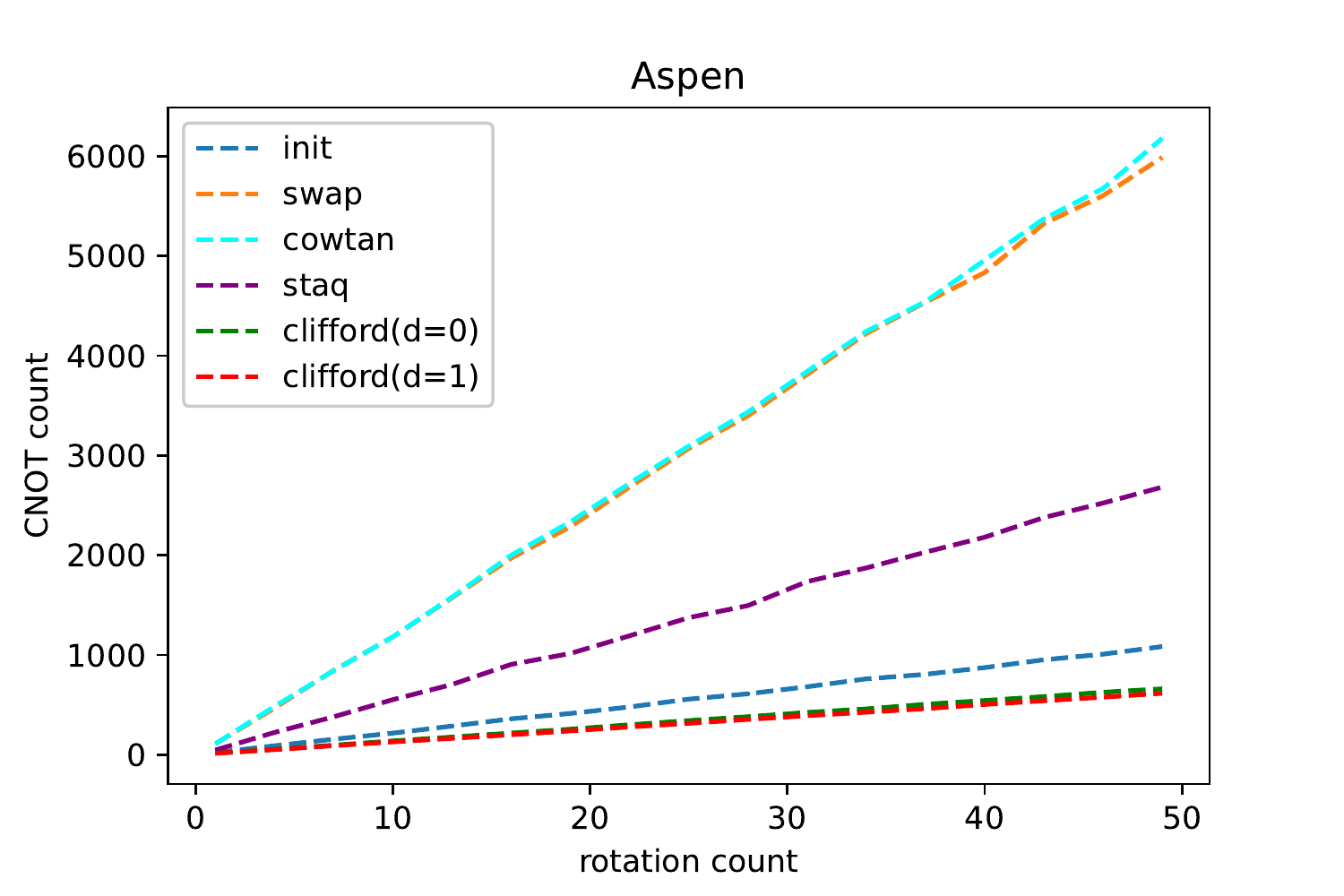}
    \includegraphics[scale=0.6]{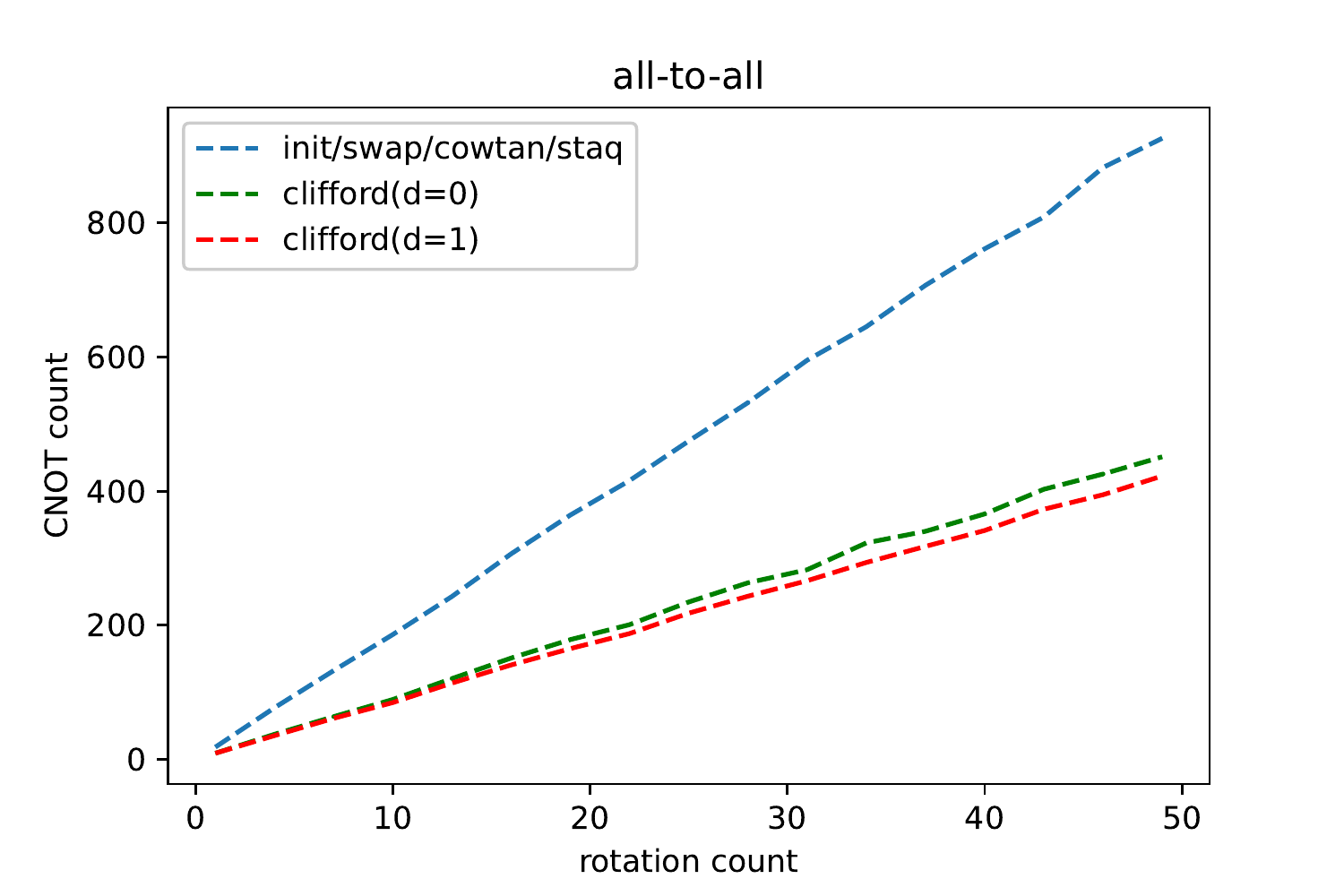}
    
    \caption{Random Pauli rotation sequence for low search depth. }
    \label{graph:pauli_low_depth}
\end{figure}

\section{Future work and possible extensions}

\subsection{Extension to other subgroups}
\noindent{\bf Beyond Clifford.} It is not clear how to extend this approach to groups larger than the Clifford group. It might be worth to investigate the higher level of the Clifford hierarchy and devise extraction routines for these of operators, even though they do not exhibit the same properties as the Clifford group.

\noindent{\bf Gaussian operators.} A potential candidate is the group of Gaussian operators. This group corresponds to operators that can be implemented via circuits of \emph{matchgates}. It has the nice feature of stabilizing the hierarchy of Hamiltonians that are bounded degree polynomials over a Clifford algebra \cite{Jozsa_2008}. It happens that these types of operators are the main ingredient used to construct UCCSD Ansätze for Fermionic dynamics (including VQE circuits for quantum chemistry or material science). For instance, in the quantum chemistry setting, this result entails that one can \emph{pull} all single excitation terms of an Ansatz to the end of the Ansatz, and conjugate the final Hamiltonian with these terms. The resulting circuit will have a reduced number of terms to implement, but these terms might be harder to implement. Hence it is not obvious that one can gain anything by synthesizing these via a naive approach.

\subsection{Changing the metric}

In this present work, we only developed algorithms that try to reduce the overall CNOT count of the output circuit. It is of course possible to change the metric to take into account more aspects of the final circuit. A good start would be to use finer hardware models and (roughly) compute the fidelity of each produced sub-circuit, picking the most faithful one. This simple approach has been proven to improve the overall circuit fidelity compared to straightforward gate count minimization in the SWAP insertion setting. Similarly one can also aim at reducing entangling depth instead of entangling gate count. 

\subsection{Global approach for synthesis of Pauli rotation sequences}

In this work, we take a local (with finite depth search) approach to tackle the problem of re-synthesis of a sequence of Pauli rotations. It could be interesting to apply global techniques for the synthesis of groups of commuting rotations to solve this problem. This would probably lead to better results for standard Clifford + T circuits. It remains unclear if these approaches can behave well for NISQ era circuits.
A recent work by Gheorghiu et al. \cite{gheorghiu2020reducing} tackles the problem by extracting and re-synthesizing phase polynomials out of the input circuit. This approach seems to perform quite well on some circuits and far worse than our method on others. For instance, Table \ref{fig:table_mosca} sums up the performances of their two splitting heuristics and our SWAP and Clifford based compilers on a $3\times 3$ grid architecture. Notice that the SWAP based compiler already outperforms their splitting techniques on such simple circuits. This difference in performances tends to shrink when compiling larger circuits which hints that their method has a far better scaling than ours.
\begin{table}[h]
    
    \caption{Compilation overhead (in CNOT gates) for some standard circuit on a $3\times 3$ grid architecture. The two right-most columns are taken from \cite{gheorghiu2020reducing}.}
    \centering
    
    \begin{tabular}{|c|c|c|c|c|c|}
\hline
circuit & init & swap &  clifford$\star\dagger$ & CNOT-OPT-A & CNOT-OPT-B\\
\hline
tof\_5 & 42 & 92.8\% & \textcolor{green}{23.8\%}  & 140.82\% & 138.78\%\\
mod\_mult\_55 & 48 & 162.5\% & \textcolor{green}{29.1\%} & 321.82\% & 203.64\%\\
barenco\_tof\_5 & 72 & 116.7\% & \textcolor{green}{-16.7\%} & 245.24\% & 140.48\%\\
grover\_5 & 288 & 89.6\% & \textcolor{green}{42.4\%} & 116.67\% & 105.36\%\\
\hline
    \end{tabular}
    \label{fig:table_mosca}
\end{table}

\section{Conclusion and discussion}\label{sec:concl}

We presented a meta-heuristic called lazy synthesis that exploits efficient representations of elements in a subgroup of the unitary group in order to compile an input quantum circuit into an architecture compliant circuit. We showed how this meta-heuristic can be used to reformulate a standard SWAP insertion algorithm from the literature and produced two new compilation algorithms based on the partial synthesis of linear Boolean operators and Clifford operators.
Finally, we ran benchmarks on various classes of circuits, providing evidence that these algorithms are competitive in a NISQ setting.

While our algorithms seems to be well behaved on NISQ oriented quantum circuits, it remains unclear of their scalability to tackle very large Clifford + T quantum circuits. It is very likely that their inherently local structure will hinder performances on large circuits.

\section*{Acknowledgments}

This work was supported in part by the French National Research Agency (ANR) under the research project SoftQPRO ANR-17-CE25-0009-02, and by
the DGE of the French Ministry of Industry under the research project PIAGDN/QuantEx P163746-484124.

\bibliographystyle{alphaurl}
\bibliography{biblio}

\newcommand{\etalchar}[1]{$^{#1}$}
\begin{thebibliography}{dBBV{\etalchar{+}}20}

\bibitem[AAM18]{amy2018controlled}
Matthew Amy, Parsiad Azimzadeh, and Michele Mosca.
\newblock On the controlled-not complexity of controlled-not--phase circuits.
\newblock {\em Quantum Science and Technology}, 4(1):015002, 2018.
\newblock \href {https://doi.org/10.1088/2058-9565/aad8ca}
  {\path{doi:10.1088/2058-9565/aad8ca}}.

\bibitem[AG04]{Aaronson_2004}
Scott Aaronson and Daniel Gottesman.
\newblock Improved simulation of stabilizer circuits.
\newblock {\em Physical Review A}, 70(5), Nov 2004.
\newblock \href {https://doi.org/10.1103/physreva.70.052328}
  {\path{doi:10.1103/physreva.70.052328}}.

\bibitem[AG20]{Amy_2020}
Matthew Amy and Vlad Gheorghiu.
\newblock staq{\textemdash}a full-stack quantum processing toolkit.
\newblock {\em Quantum Science and Technology}, 5(3):034016, jun 2020.
\newblock \href {https://doi.org/10.1088/2058-9565/ab9359}
  {\path{doi:10.1088/2058-9565/ab9359}}.

\bibitem[BM20]{bravyi2020hadamard}
Sergey Bravyi and Dmitri Maslov.
\newblock Hadamard-free circuits expose the structure of the clifford group.
\newblock {\em arXiv preprint arXiv:2003.09412}, 2020.
\newblock \href {https://doi.org/10.1109/TIT.2021.3081415}
  {\path{doi:10.1109/TIT.2021.3081415}}.

\bibitem[CDD{\etalchar{+}}19]{cowtan2019qubit}
Alexander Cowtan, Silas Dilkes, Ross Duncan, Alexandre Krajenbrink, Will
  Simmons, and Seyon Sivarajah.
\newblock On the qubit routing problem.
\newblock In {\em 14th Conference on the Theory of Quantum Computation,
  Communication and Cryptography (TQC 2019)}. Schloss Dagstuhl-Leibniz-Zentrum
  fuer Informatik, 2019.
\newblock \href {https://doi.org/10.4230/LIPIcs.TQC.2019.5}
  {\path{doi:10.4230/LIPIcs.TQC.2019.5}}.

\bibitem[CSU19]{Childs2019CircuitTF}
Andrew~M. Childs, E.~Schoute, and Cem~M. Unsal.
\newblock Circuit transformations for quantum architectures.
\newblock {\em ArXiv}, abs/1902.09102, 2019.
\newblock \href {https://doi.org/10.4230/LIPIcs.TQC.2019.3}
  {\path{doi:10.4230/LIPIcs.TQC.2019.3}}.

\bibitem[dB11]{beaudrap2011linearized}
Niel de~Beaudrap.
\newblock A linearized stabilizer formalism for systems of finite dimension.
\newblock 2011.
\newblock \href {https://doi.org/10.5555/2481591.2481597}
  {\path{doi:10.5555/2481591.2481597}}.

\bibitem[dBBV{\etalchar{+}}20]{de2020quantum}
Timoth{\'e}e~Goubault de~Brugi{\`e}re, Marc Baboulin, Beno{\^\i}t Valiron,
  Simon Martiel, and Cyril Allouche.
\newblock Quantum cnot circuits synthesis for nisq architectures using the
  syndrome decoding problem.
\newblock In {\em International Conference on Reversible Computation}, pages
  189--205. Springer, 2020.
\newblock \href {https://doi.org/10.1007/978-3-030-52482-1_11}
  {\path{doi:10.1007/978-3-030-52482-1_11}}.

\bibitem[GLMM20]{gheorghiu2020reducing}
Vlad Gheorghiu, Sarah~Meng Li, Michele Mosca, and Priyanka Mukhopadhyay.
\newblock Reducing the cnot count for clifford+t circuits on nisq
  architectures, 2020.
\newblock \href {http://arxiv.org/abs/2011.12191} {\path{arXiv:2011.12191}}.

\bibitem[HNYN11]{hirata2011}
Yuichi Hirata, Masaki Nakanishi, Shigeru Yamashita, and Yasuhiko Nakashima.
\newblock An efficient conversion of quantum circuits to a linear nearest
  neighbor architecture.
\newblock {\em Quantum Information \& Computation}, 11(1\&2):142--166, 2011.
\newblock \href {https://doi.org/10.26421/QIC11.1-2-10}
  {\path{doi:10.26421/QIC11.1-2-10}}.

\bibitem[JM08]{Jozsa_2008}
Richard Jozsa and Akimasa Miyake.
\newblock Matchgates and classical simulation of quantum circuits.
\newblock {\em Proceedings of the Royal Society A: Mathematical, Physical and
  Engineering Sciences}, 464(2100):3089–3106, Jul 2008.
\newblock \href {https://doi.org/10.1098/rspa.2008.0189}
  {\path{doi:10.1098/rspa.2008.0189}}.

\bibitem[KMS07]{kutin2007computation}
Samuel~A Kutin, David~Petrie Moulton, and Lawren~M Smithline.
\newblock Computation at a distance.
\newblock {\em arXiv preprint quant-ph/0701194}, 2007.
\newblock \href {https://doi.org/10.48550/arXiv.quant-ph/0701194}
  {\path{doi:10.48550/arXiv.quant-ph/0701194}}.

\bibitem[KvdG19]{kissinger2019cnot}
Aleks Kissinger and Arianne~Meijer van~de Griend.
\newblock Cnot circuit extraction for topologically-constrained quantum
  memories, 2019.
\newblock \href {https://doi.org/10.26421/QIC20.7-8-4}
  {\path{doi:10.26421/QIC20.7-8-4}}.

\bibitem[LDX18]{sabre}
Gushu Li, Yufei Ding, and Yuan Xie.
\newblock Tackling the qubit mapping problem for nisq-era quantum devices,
  2018.
\newblock \href {https://doi.org/10.1145/3297858.3304023}
  {\path{doi:10.1145/3297858.3304023}}.

\bibitem[Lit19]{Litinski_2019}
Daniel Litinski.
\newblock Magic state distillation: Not as costly as you think.
\newblock {\em Quantum}, 3:205, Dec 2019.
\newblock \href {https://doi.org/10.22331/q-2019-12-02-205}
  {\path{doi:10.22331/q-2019-12-02-205}}.

\bibitem[NGM20]{Nash_2020}
Beatrice Nash, Vlad Gheorghiu, and Michele Mosca.
\newblock Quantum circuit optimizations for nisq architectures.
\newblock {\em Quantum Science and Technology}, 5(2):025010, Mar 2020.
\newblock \href {https://doi.org/10.1088/2058-9565/ab79b1}
  {\path{doi:10.1088/2058-9565/ab79b1}}.

\bibitem[PMH08]{patel2008optimal}
Ketan~N Patel, Igor~L Markov, and John~P Hayes.
\newblock Optimal synthesis of linear reversible circuits.
\newblock {\em Quantum Information \& Computation}, 8(3):282--294, 2008.
\newblock \href {https://doi.org/10.5555/2011763.2011767}
  {\path{doi:10.5555/2011763.2011767}}.

\bibitem[Pre18]{Preskill_2018}
John Preskill.
\newblock Quantum computing in the nisq era and beyond.
\newblock {\em Quantum}, 2:79, Aug 2018.
\newblock \href {https://doi.org/10.22331/q-2018-08-06-79}
  {\path{doi:10.22331/q-2018-08-06-79}}.

\bibitem[SSP13]{6560634}
A.~{Shafaei}, M.~{Saeedi}, and M.~{Pedram}.
\newblock Optimization of quantum circuits for interaction distance in linear
  nearest neighbor architectures.
\newblock In {\em 2013 50th ACM/EDAC/IEEE Design Automation Conference (DAC)},
  pages 1--6, 2013.
\newblock \href {https://doi.org/10.1145/2463209.2488785}
  {\path{doi:10.1145/2463209.2488785}}.

\bibitem[Tak90]{takahashi1990approximate}
Hiromitsu Takahashi.
\newblock An approximate solution for the steiner problem in graphs.
\newblock {\em Math. Japonica.}, 6:573--577, 1990.

\bibitem[vdBT20]{van_den_Berg_2020}
Ewout van~den Berg and Kristan Temme.
\newblock Circuit optimization of hamiltonian simulation by simultaneous
  diagonalization of pauli clusters.
\newblock {\em Quantum}, 4:322, Sep 2020.
\newblock \href {https://doi.org/10.22331/q-2020-09-12-322}
  {\path{doi:10.22331/q-2020-09-12-322}}.

\bibitem[vdGD20]{griend2020architectureaware}
Arianne~Meijer van~de Griend and Ross Duncan.
\newblock Architecture-aware synthesis of phase polynomials for nisq devices,
  2020.
\newblock \href {https://doi.org/10.4204/EPTCS} {\path{doi:10.4204/EPTCS}}.

\bibitem[ZC19]{zhang2019optimizing}
Fang Zhang and Jianxin Chen.
\newblock Optimizing t gates in clifford+t circuit as $\pi/4$ rotations around
  paulis, 2019.
\newblock \href {http://arxiv.org/abs/1903.12456} {\path{arXiv:1903.12456}},
  \href {https://doi.org/10.48550/arXiv.1903.12456}
  {\path{doi:10.48550/arXiv.1903.12456}}.

\bibitem[ZPW17]{bka}
Alwin Zulehner, Alexandru Paler, and Robert Wille.
\newblock An efficient methodology for mapping quantum circuits to the ibm qx
  architectures, 2017.
\newblock \href {https://doi.org/10.1109/TCAD.2018.2846658}
  {\path{doi:10.1109/TCAD.2018.2846658}}.

\end{thebibliography}

\appendix
\section{Dealing with the final operator}\label{sec:opt}

In this section we detail how to classically emulate any final non-trivial Clifford operator.
This encompasses the case of a final permutation or linear operator, even though the case of a final permutation can be trivially dealt with. 

\subsection{Expected value of some observable}
In this setting, we assume that we are given as input, both the circuit $C_{in}$ to compile and some final observable $H$ to evaluate and the end of the circuit execution. In short, we need to compute:
$$ \bra{0}C_{in}^\dagger H C_{in} \ket{0} $$
Using either the linear operator synthesis approach of the Clifford approach, we end up producing a circuit $C_{out}$ and a final linear/Clifford operator $A$ such that:

$$ \bra{0}C_{in}^\dagger H C_{in} \ket{0} = \bra{0}C_{out}^\dagger A^\dagger H A C_{out} \ket{0}  $$

Lets further assume that $H$ is given to us in the Pauli basis. That is:
$$ H = \sum_i \alpha_i P_i $$
with $\alpha_i$ some real coefficients, and $P_i \in \mathcal{P}_n$ some Pauli operators.
Then, sampling the new observable $A^\dagger H A = \sum_i \alpha_i A^\dagger P_i A = \sum_i \alpha_i P_i'$ on the output circuit is equivalent to sampling the input observable on the input circuit.

Sampling this observable using the standard techniques of co-diagonalization of its terms is no more costly (in terms of shots) than sampling the original $H$.

\subsection{Sampling bit-strings}
If we are required to provide some samples taken according to the final distribution induced by $C_{in}\ket{0}$, things are bit trickier.

Our algorithms output a pair $C_{out}, A$ such that $A C_{out} = C_{in}$ and we would like to emulate sampling of $AC_{out}\ket{0} = C_{in}\ket{0}$. To do so we proceed as follows.

Defining $\mathcal{Z} =\{Z_i,\ i\in [n]\}$ the set of local $Z$ operators on each qubit. Sampling bit-strings out of some quantum state over $n$ qubits boils down to iteratively evaluating the value of these operators in any order (since they commute). We would like to evaluate this collection of operators on state $AC_{out}\ket{0}$. This is equivalent to evaluating the collection of Pauli operators $A^\dagger \mathcal{Z} A = \{P_i = A^\dagger Z_i A,\ i\in [n]\}$ on state $C_{out}\ket{0}$. These operators, however, might not be diagonal operators, and thus cannot be directly evaluated using standard computational basis measurements. Nevertheless, these operators commute with one another since the $Z_i$ commute. Hence, one can co-diagonalize them using a Clifford circuit $C_{diag}$. By construction, the new collection of Pauli operators $\{s_i Q_i = C_{diag}P_i C_{diag}^\dagger,\ i \in [n]\}$ are diagonal operators (hence products of $Z_j$ and $I_j$) times some phase $s_i = \pm 1$. Sampling these operators on state $C_{diag} C_{out}\ket{0}$ is, by construction, equivalent to evaluating operators in $\mathcal{Z}$ over state  $A C_{out}\ket{0}$. Figure \ref{fig:codiag} depicts this sequence of conjugations.

Since the $Q_i$ are products of $Z_j$ and $I_j$ operators, they can be seen as computing parities over a subset of qubits. This gives us a simple algorithm to fix measurement results. We can sample some bit-string $w$ out of the quantum state $C_{diag} C_{out}\ket{0}$ and output a new bit-string $w'$ with $w'_i =\delta_{s_i}^{-1} \oplus \sum_{j\in Q_i} w_j$ where the sum is modulo 2. This operation boils down to applying an affine system over $\mathbb{F}_2^n$ described by the $(s_i, Q_i)$ operators.

For example, let's assume that we need to sample bit-strings over 2 qubits. Let's assume that after conjugation through $A$ and co-diagonalization, we get operators $Q_1 = Z \otimes Z$, $s_1 = -1$ and $Q_2 = Z \otimes I$, $s_2= 1$. These operators can be summed up via the following affine system over $\mathbb{F}_2^2$:
$$ x \mapsto Lx + b $$
with
$L = \begin{pmatrix}
1 & 1\\
1 & 0
\end{pmatrix}$
and $b=(1, 0)^T$. Any bit-string sampled from state $C_{diag} C_{out}\ket{0}$ can be fixed by applying $L$ and adding $b$:
\begin{align*}
    00 \mapsto 10\\
    01 \mapsto 00\\
    10 \mapsto 01\\
    11 \mapsto 11\\
\end{align*}

\begin{figure}[h]
\centering
(a)\Qcircuit @C=0.9em @R=0.7em {
&&\multigate{4}{C_{out}}&\multigate{4}{A}&\gate{Z_1} &\qw\\
&&\ghost{C_{out}}       &\ghost{A}       &\gate{Z_2}&\qw\\
&&\ghost{C_{out}}       &\ghost{A}       &\vdots&\\
&&\ghost{C_{out}}       &\ghost{A}       &\vdots &\\
&&\ghost{C_{out}}       &\ghost{A}       &\gate{Z_n}&\qw
}~~(b)\Qcircuit @C=0.9em @R=0.7em {
& \multigate{4}{C_{out}}&\multigate{4}{P_1} &\multigate{4}{P_2} &\cdots&& \multigate{4}{P_n} &\multigate{4}{A}&\\
 & \ghost{C_{out}}       &\ghost{P_1}        &\ghost{P_2}        &&      & \ghost{P_n}        &\ghost{A} &\\
 & \ghost{C_{out}}          &\ghost{P_1}        &\ghost{P_2}        &\cdots&& \ghost{P_n}        &\ghost{A} &\\
  & \ghost{C_{out}}          &\ghost{P_1}        &\ghost{P_2}        &&      & \ghost{P_n}        &\ghost{A} &\\
 & \ghost{C_{out}}       &\ghost{P_1}        &\ghost{P_2}        &\cdots&& \ghost{P_n}        &\ghost{A} &
}

(c)\Qcircuit @C=0.9em @R=0.7em {
 & \multigate{4}{C_{out}}&\multigate{4}{C_{diag}}&\multigate{4}{s_1Q_1}&\multigate{4}{s_2Q_2}&\cdots&&\multigate{4}{s_nQ_n}&\multigate{4}{C_{diag}^\dagger}&\multigate{4}{A}&\\
 & \ghost{C_{out}}       &\ghost{C_{diag}}       &\ghost{s_1Q_1}       &\ghost{s_2Q_2}       &&      &\ghost{s_nQ_n}       &\ghost{C_{diag}^\dagger}       &\ghost{A} &\\
 & \ghost{C_{out}}          &\ghost{C_{diag}}       &\ghost{s_1Q_1}       &\ghost{s_2Q_2}       &\cdots&&\ghost{s_nQ_n}       &\ghost{C_{diag}^\dagger}       &\ghost{A} &\\
 & \ghost{C_{out}}          &\ghost{C_{diag}}       &\ghost{s_1Q_1}       &\ghost{s_2Q_2}       &&      &\ghost{s_nQ_n}       &\ghost{C_{diag}^\dagger}       &\ghost{A} &\\
 & \ghost{C_{out}}       &\ghost{C_{diag}}       &\ghost{s_1Q_1}       &\ghost{s_2Q_2}       &\cdots&&\ghost{s_nQ_n}       &\ghost{C_{diag}^\dagger}       &\ghost{A} &
}
\caption{The sampling fixing procedure. (a) we need to emulate sampling of the quantum state $AC_{out}\ket{0} = C_{in}\ket{0}$. This sampling procedure relies on the joint measurements of operators $Z_i$ for each qubit $i$. (b) Since $A$ is Clifford, we can commute the $Z_i$ with $A$, yielding a collection of commuting operators $P_i$ (non necessarily diagonal). (c) These operators can be jointly measured by co-diagonalizing them via a Clifford circuit $C_{diag}$, yielding a collection of diagonal operators $s_iQ_i$ where $Q_i$ are products of $Z$ operators and $s_i = \pm 1$ are phases. Once the $s_i Q_i$ are measured a simple linear system inversion allows us to emulate the sampling of the initial $Z_i$ operators. Hence, in practice, only $C_{out}$ and $C_{diag}$ are effectively performed on the quantum processor.}\label{fig:codiag}
\end{figure}
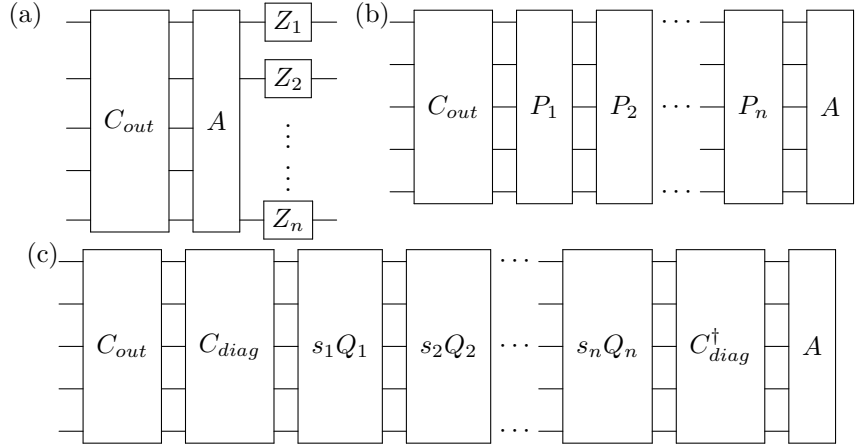

\noindent{\bf Remark on co-diagonalization.} To the best of our knowledge, \cite{van_den_Berg_2020} provides the best approach to produce a co-diagonalization circuit. They show that this task can be reduced to the synthesis of a linear Boolean operator which has a worst case complexity of $O(n^2/\log{n})$ in the case of non constrained architecture. A simple application of any architecture-aware CNOT synthesis heuristics, like \cite{de2020quantum}, gives an architecture-aware algorithm to produce $C_{diag}$. This argument is essential to claim that, in most of the cases, it is more efficient to synthesize $C_{diag}$ rather than directly synthesizing $A$. Indeed, the synthesis of an arbitrary Clifford operator is usually done by the synthesis of successive layers of Hadamard gates, Phase gates, CNOT gates or CZ gates. Most recent results show that three layers of two-qubit gates are necessary \cite{bravyi2020hadamard}, making it more affordable to use the co-diagonalization process.

\section{Recursive search of finite depth}\label{sec:rec_search}

In the three algorithms described in this paper, the extraction functions perform a recursive search over the next $w$ calls to the extraction in order to locally pick the subcircuit that will generate the least extraction overhead. This recursive search is introduced in \cite{hirata2011} for SWAP insertion and can be easily transposed to the linear Boolean operator and Clifford setting.

In practice, this is done by computing a tree of depth $w$ containing all the possible choices and associating to each leaf of this tree the sum of all sub-circuits scores on the path from the root to the leaf. In the algorithms presented in this paper, we simply used the CNOT count as a metric, but any other metric, such as overall fidelity, or depth can be used in this search.
Figure \ref{fig:tree_depth} depicts two search trees of depth 0 and 1. 

\begin{figure}[h]
   
\begin{center}

\begin{tikzpicture}
\node[rectangle, draw, inner sep=3pt](root) at (0, 0){+0};

\node[rectangle, draw, inner sep=3pt](n0) at (-2, -2){+6};
\node[rectangle, draw, inner sep=3pt](n3) at (0, -2){+6};
\node[rectangle, draw, inner sep=3pt](n4) at (2, -2){+6};
\draw[thick, blue] (root) edge node[left, circle, inner sep=1pt, draw, xshift=-0.2cm]{a} (n0);
\draw (root) edge node[left, circle, inner sep=1pt, draw, xshift=-0.1cm]{b} (n3);
\draw (root) edge node[right, circle, inner sep=1pt, draw, xshift=0.2cm]{c} (n4);

\node[rectangle, fill=green!50] at (-2, -3){+6};
\node[rectangle, fill=green!50] at (0, -3){+6};
\node[rectangle, fill=green!50] at (2, -3){+6};

\end{tikzpicture}\hspace{1cm}\begin{tikzpicture}
\node[rectangle, draw, inner sep=3pt](root) at (0, 0){+0};

\node[rectangle, draw, inner sep=3pt](n0) at (-2, -2){+6};
\node[rectangle, draw, inner sep=3pt](n3) at (0, -2){+6};
\node[rectangle, draw, inner sep=3pt](n4) at (2, -2){+6};
\draw (root) edge node[left, circle, inner sep=1pt, draw, xshift=-0.2cm]{$a$} (n0);
\draw[thick, blue] (root) edge node[left, circle, inner sep=1pt, draw, xshift=-0.1cm]{$b$} (n3);
\draw (root) edge node[right, circle, inner sep=1pt, draw, xshift=0.2cm]{$c$} (n4);

\node[rectangle, draw, inner sep=3pt](n00) at (-3, -4){+6};
\node[rectangle, draw, inner sep=3pt](n03) at (-2, -4){+6};
\draw (n0) edge node[left, circle, draw, inner sep=1pt]{$a$} (n00);
\draw (n0) edge node[right, circle, draw, inner sep=1pt]{$b$}(n03);

\node[rectangle, draw, inner sep=3pt](n30) at (0, -4){+3};
\draw (n3) edge node[left, circle, draw, inner sep=1pt]{$a$} (n30);

\node[rectangle, draw, inner sep=3pt](n43) at (2, -4){+6};
\node[rectangle, draw, inner sep=3pt](n44) at (3, -4){+8};
\draw (n4) edge node[left, circle, draw, inner sep=1pt]{$a$} (n43);
\draw (n4) edge node[right, circle, draw, inner sep=1pt]{$b$}(n44);

\node[rectangle, fill=red!50] at (-3, -5){+12};
\node[rectangle, fill=red!50] at (-2, -5){+12};
\node[rectangle, fill=green!50] at (0, -5){+9};
\node[rectangle, fill=red!50] at (2, -5){+12};
\node[rectangle, fill=red!50] at (3, -5){+14};

\end{tikzpicture}

\end{center}
    \caption{Search trees of depth 0 and 1. In the first tree, we stop the recursive search at depth 0. In this situation all possible choices {\protect\tikz \protect\draw node[circle, draw, inner sep=1]{$a$}; }, {\protect\tikz \protect\draw node[circle, draw, inner sep=1]{$b$}; }, and {\protect\tikz \protect\draw node[circle, draw, inner sep=1pt]{$c$}; } are equivalent since they produce sub-circuits of score $6$. Hence, we greedily pick the first choice {\protect\tikz \protect\draw node[circle, draw, inner sep=1]{$a$}; }. After exploring at depth $1$, we notice that choosing option {\protect\tikz \protect\draw node[circle, draw, inner sep=1]{$b$}; } (in blue) is the best local choice since it lead to a sub-tree with the least costly leaf (the middle leaf with cost $9$).}
    \label{fig:tree_depth}
\end{figure}

\section{Detailed examples}
This Appendix regroups detailed step-by-step runs of the three algorithms presented in this paper.

The circuit we will compile is the following:

\begin{tikzpicture}
\tikzstyle{oneqb} = [draw,fill=white,minimum size=1.5em]
\tikzstyle{swp} = [minimum size=7pt, inner sep=3pt, path picture={\draw[thick] (path picture bounding box.north west) -- (path picture bounding box.south east)(path picture bounding box.north east) -- (path picture bounding box.south west);}]
\tikzstyle{ctrl} = [fill,shape=circle,minimum size=5pt,inner sep=0pt]
\tikzstyle{targ} = [circle, draw, minimum size=7pt, inner sep=3pt, path picture={\draw[thick] (path picture bounding box.north) -- (path picture bounding box.south)(path picture bounding box.east) -- (path picture bounding box.west);}]
\coordinate (anchor_0) at (0.5, -0);
\coordinate (anchor_1) at (0.5, -1);
\coordinate (anchor_2) at (0.5, -2);
\coordinate (anchor_3) at (0.5, -3);
\coordinate (anchor_4) at (0.5, -4);
\coordinate (anchor_5) at (0.5, -5);
\draw (0, 0) node[](q0){$q_{0}$};
\draw (0, -1) node[](q1){$q_{1}$};
\draw (0, -2) node[](q2){$q_{2}$};
\draw (0, -3) node[](q3){$q_{3}$};
\draw (0, -4) node[](q4){$q_{4}$};
\draw (0, -5) node[](q5){$q_{5}$};
\draw (1.1, 0) node[oneqb](g0){\tiny H};
\draw (1.1, -4) node[oneqb](g1){$\sqrt{X}$};
\draw (2.2, 0) node[ctrl](c2){};
\draw (2.2, -4) node[targ](t2){};
\draw (3.3000000000000003, -4) node[ctrl](c3){};
\draw (3.3000000000000003, -2) node[targ](t3){};
\draw (4.4, -2) node[oneqb](g4){\tiny T};
\draw (5.5, -2) node[oneqb](g5){\tiny H};
\draw (3.3000000000000003, -1) node[oneqb](g6){$\sqrt{X}$};
\draw (6.6, -1) node[ctrl](c7){};
\draw (6.6, -5) node[targ](t7){};
\draw (7.699999999999999, -5) node[ctrl](c8){};
\draw (7.699999999999999, -3) node[targ](t8){};
\draw (8.799999999999999, -3) node[oneqb](g9){\tiny T};
\draw (9.799999999999999, 0) node[](e0){};
\draw (9.799999999999999, -1) node[](e1){};
\draw (9.799999999999999, -2) node[](e2){};
\draw (9.799999999999999, -3) node[](e3){};
\draw (9.799999999999999, -4) node[](e4){};
\draw (9.799999999999999, -5) node[](e5){};
\draw[thick] (q0-|anchor_0.east) -- (q0-|g0.west);
\draw[thick] (q4-|anchor_4.east) -- (q4-|g1.west);
\draw[thick] (q0-|g0.east) -- (q0-|c2.west);
\draw[thick] (q4-|g1.east) -- (q4-|t2.west);
\draw[thick] (t2) -- (c2);
\draw[thick] (q4-|t2.east) -- (q4-|c3.west);
\draw[thick] (q2-|anchor_2.east) -- (q2-|t3.west);
\draw[thick] (t3) -- (c3);
\draw[thick] (q2-|t3.east) -- (q2-|g4.west);
\draw[thick] (q2-|g4.east) -- (q2-|g5.west);
\draw[thick] (q1-|anchor_1.east) -- (q1-|g6.west);
\draw[thick] (q1-|g6.east) -- (q1-|c7.west);
\draw[thick] (q5-|anchor_5.east) -- (q5-|t7.west);
\draw[thick] (t7) -- (c7);
\draw[thick] (q5-|t7.east) -- (q5-|c8.west);
\draw[thick] (q3-|anchor_3.east) -- (q3-|t8.west);
\draw[thick] (t8) -- (c8);
\draw[thick] (q3-|t8.east) -- (q3-|g9.west);
\draw[thick] (q0-|c2.east) -- (q0-|e0.west);
\draw[thick] (q1-|c7.east) -- (q1-|e1.west);
\draw[thick] (q2-|g5.east) -- (q2-|e2.west);
\draw[thick] (q3-|g9.east) -- (q3-|e3.west);
\draw[thick] (q4-|c3.east) -- (q4-|e4.west);
\draw[thick] (q5-|c8.east) -- (q5-|e5.west);
\end{tikzpicture}

In all examples, we will compile for a linear-nearest-neighbor architecture with $6$ qubits.

For circuits displayed below, gates are red if they do not belong to the target subgroup $S$ and blue if they are in $S$. The gate that is currently being treated by the algorithm and the corresponding extracted subcircuit are displayed in green.

\subsection{Permutation based lazy synthesis}
\label{app:example_swap}

The execution of the algorithm is depicted as follows. 
Each row contains: the current output circuit, what is left of the input circuit (including the last treated gate, in green), the current permutation $\sigma$, the target qubits of the gate being extracted (if it is an extraction step) denoted $g$, and its mapping through $\sigma^{-1}$. When the gate is not compatible with the architecture, this mapping is displayed in red, and the extracted permutation $\pi$ and the new updated permutation $\sigma'$ are displayed.

\input{examples/example_swaps.tex}

\subsection{Linear reversible based lazy synthesis}
\label{app:example_linear}

Each row contains: the current output circuit, what is left of the input circuit (including the last treated gate, in green), the inverse of the current table $A^{-1}$.
If this is an update step, two columns of the table will be highlighted, corresponding to the two columns that were summed. The green column is summed into the cyan one.
If this is an extraction step, a row is highlighted corresponding to the row that was synthesized by the fan-in step. If the gate is non diagonal, a column is also highlighted corresponding to the column synthesized by the fan-out step.
Notice that the highlighted row corresponds to the index of the receiving qubits in the output circuit.

\begin{center}
\begin{adjustwidth}{-1.5cm}{}
\input{examples/example_cnots.tex}
\end{adjustwidth}
\end{center}

\subsection{Clifford based lazy synthesis}
\label{app:example_clifford}
Each row contains: the current output circuit, what is left of the input circuit (including the last treated gate, in green), the current Tableau data structure (or its inverse for extraction steps). The Tableau is represented as a $2n\times n$ array, even columns $2i$ represent the pauli operator $C X_i C^\dagger$, while odd columns $2i+1$ represent pauli operator $C Z_i C^\dagger$. For extraction step, the inverse Tableau is represented. This entails that, upon meeting a $T$ gate (the
only non Clifford gates in this example) on qubit $q$, one just need to read column $2q+1$ in order to know which half rotation needs to be implemented.
\begin{center}
\begin{adjustwidth}{-1.5cm}{}
\input{examples/example_clifford.tex}
\end{adjustwidth}
\end{center}
\end{document}